%% file: arxiv_Decomp.tex
\documentclass{article}

\RequirePackage[OT1]{fontenc}
\RequirePackage{amsmath,amsfonts,amsthm,amssymb}
\usepackage{enumerate}
\usepackage{framed}
\usepackage{authblk}
\usepackage[round]{natbib}

\numberwithin{equation}{section}

\newtheorem{theorem}{Theorem}

\newtheorem{lemma}{Lemma}
\newtheorem{proposition}{Proposition}
\newtheorem{remark}{Remark}

\newtheorem{example}{Example}

\usepackage{graphicx}
\usepackage[a4paper,top=2cm, bottom=2cm, left=2cm, right=2cm]{geometry}

\newcommand{\bP}{\mathbb{P}}

\newcommand{\bE}{\mathbb{E}}
\newcommand{\bR}{\mathbb{R}}

\newcommand{\mA}{\mathcal{A}}
\newcommand{\mB}{\mathcal{B}}
\newcommand{\mC}{\mathcal{C}}
\newcommand{\mF}{\mathcal{F}}
\newcommand{\mI}{\mathcal{I}}
\newcommand{\mJ}{\mathcal{J}}
\newcommand{\mL}{\mathcal{L}}
\newcommand{\mK}{\mathcal{K}}
\newcommand{\mP}{\mathcal{P}}
\newcommand{\mS}{\mathcal{S}}
\newcommand{\mT}{\mathcal{T}}
\newcommand{\mV}{\mathcal{V}}

\newcommand{\bI}{{\mathbf I}}

\newcommand{\T}{\textsf{T}}

\begin{document}
 
 \title{Efficient Advert Assignment}
 
\author[1]{Frank Kelly}
\affil[1]{University of Cambridge}
 
\author[2]{Peter Key}
\affil[2]{Microsoft Research Cambridge}

\author[3]{Neil Walton}
\affil[3]{University of Amsterdam}
\date{}

\maketitle

\begin{abstract}
We develop a framework for the analysis of  large-scale Ad-auctions where
adverts are assigned over a continuum of search types. 
For this pay-per-click market, we provide an efficient 
mechanism that maximizes social welfare. 
In particular, we show that the social welfare optimization can be solved in separate optimizations conducted on the time-scales relevant to the 
search platform and advertisers. Here, on each search occurrence, the
platform solves an assignment problem and, on a slower 
time-scale, each advertiser submits a bid which matches its demand for click-throughs with supply. Importantly, knowledge of global parameters, such as the distribution of search terms, is not required when separating the problem in this way. 
Exploiting the information asymmetry between the platform and advertiser, we describe 
a simple mechanism which incentivizes truthful bidding and has a unique
Nash equilibrium that is socially optimal, and thus implements our
decomposition. Further, we consider models where advertisers adapt their
bids smoothly over time, and prove convergence to the solution that
maximizes social welfare.
Finally, we describe several extensions which illustrate the flexibility and tractability of our framework.
\end{abstract}

\let\clearpage\relax

\input{INTRODUCTION_FINAL}

\input{DECOMPOSITION_FINAL}

\input{MECH_DESIGN_FINAL}

\input{DYNAMICS_FINAL}

\input{EXTENSIONS_FINAL}

\input{RELATED_FINAL}

\input{CONCLUDE_FINAL}

\bibliographystyle{plainnat}
\bibliography{REFERENCES}

\appendix
\input{APPENDIX_FINAL}

\end{document}

%% file: INTRODUCTION_FINAL.tex

\section{Introduction}

Ad-auctions lie at the heart of search markets and generate billions of
dollars in revenue for
platforms such as Bing and Google. Sponsored search auctions provide a
distributed mechanism
where advertisers compete for their adverts to be shown to users of the
search platform, by bidding
on search terms associated with queries. 

The earliest
search auction\footnote{initiated in 1998 by GoTo.com, which later became 
Overture and then part of Yahoo!} required that advertisers bid
a separate price to place an advert in each position on the search page. 
This design was soon abandoned for one where an advertiser 
simply bid an amount
per click: this amount was converted to adjusted bids for each
position by 
multiplication by the platform's estimate of
click-through probabilities; the highest adjusted bid 
won the first position, the second-highest the second position, 
and so on, with payments only made when an advert was clicked. 
The shift from an advertiser making a separate bid for each position
to the advertiser making a single bid and being charged per click
is an example of \textit{conflation}~(\cite{milgrom2010simplified}): advertisers are required to make
the same bid per click whatever the position of the
advert.

The design used to assign adverts 
to positions on the page and the rules used to determine payments 
have changed several times, with platforms such as BingAds or
Google AdWords  using variants of the generalised second-price
auction (GSP) to determine the price per click: 
under GSP 
the amount an advertiser
pays when its advert is actually clicked is the smallest price per click
that, if bid, would have won the same advert position~(\cite{varian07, eos07}).

In current auctions
a fundamental information asymmetry between
the platform and  advertisers has emerged, in that the platform  typically knows more 
than an advertiser about the search being conducted.
For example, information
on the user conducting the search 
may comprise location, previous search history,
or personal information provided by the user on sign-in
to a platform, any of which may affect click-through probabilities.  
The keyword and additional query information all vary randomly with a
distribution that is, in principle, unknown to the platform and advertisers.
However, the platform can
choose prices and an allocation of adverts to positions 
using the platform's additional
information.  In contrast to the platform,  the advertiser has to rely on  more coarse-grained
information, perhaps just the keywords of a query together with
a crude categorization of the user.
At best an advertiser sees censored information conditional on her advert
being shown and clicked: the advertiser has no information about auctions
where her advert was 
either not shown (a losing auction for her) or not clicked,
unless the platform chooses to reveal such information.

Variability in the platform's additional information
creates additional 
variability in the observations available to the advertiser.
It is difficult 
for the advertiser
to view consecutive allocations by the platform as repeated instances of the
same auction since, even if the keyword and range
of competitors stay the same, the platform's additional
information varies from search to search. 
For example, two different searchers for the same keyword may have very
different preferences for adverts, giving different click-through 
probabilities
and thus different auctions. These auctions are conflated,
with the same bid from an advertiser used in each of them,
and this conflation is additional to the conflation over positions. 
The observations available to the advertiser are inherently
stochastic, with the probability of a click-through 
fluctuating from search to search, and 
need to be filtered in order to estimate click-through rates. 
Thus the information asymmetry
between platform and advertisers induces a temporal asymmetry:  
the platform observes each search as it happens, whereas the advertiser has
to rely on delayed and aggregated feedback.    
At the time of writing,   platforms typically provide delayed feedback 
to the advertiser on quantities such as number of impressions (i.e., 
appearances), average position, average cost per click and so on, averaged over some interval of time.

In this paper we develop a framework to address the 
information and temporal asymmetries directly.  
The framework is described in~Section~\ref{sec:assign}
and operates as follows.   Each advertiser submits a real-valued bid.  
Using the 
advertisers' bids  and the platform's 
estimated click-through probabilities, the platform 
assigns adverts to maximize the (expected) bid of an advert receiving a
click. This is a classical assignment problem and is solved with low
computational overhead on each search instance. We price adverts with a
form of parametrized VCG-payment: if an advert receives a
click-through, then the pricing of that advert requires one further
solution to the assignment problem. 
We model advertiser $i$ as a 
utility maximizer,  who wants to maximize her 
payoff
$u_i(y_i)=U_i(y_i)- \pi_i y_i$, where $U_i$ is her private, concave,
utility function
and where
$\pi_i$ and $y_i$ are, respectively, the (average) price she pays per click
and the click-through rate she achieves.  Under a monotonicity
condition which will be satisfied if the platform's additional
information is sufficiently fine-grained, we 
prove that the Nash equilibrium achieved by advertisers maximizes the
social welfare of the advertisers.
The framework is extended, in Section~\ref{sec:ext}, 
to allow an advertiser to make 
different bids for different 
keywords or categories of user. 

Optimization frameworks of this form are well-established in the communication network community, where users, the network and its components must be separated.  There the phrase ``Network Utility Maximization'' has been coined, but this framework has only recently found its way into Mechanism Design (see \cite{maheswaran2004social}, \cite{YaHa07} and \cite{jots09}).   By contrast, much of the existing literature on sponsored search has needed to restrict
attention to an isolated instance  of an auction (a single query, 
repeated without variation)  to make progress.   
We focus upon the stream
of search queries: their randomness and the resulting information asymmetry
is an intrinsic aspect of our framework.

Implementation of a Nash equilibrium in the economics literature
is typically based on the assumption of complete information.
In the context of sponsored search, where an advertiser
is bidding in a conflated set of auctions with little information
on users or competitors (see \cite{pk:11}), the complete information
assumption is not compelling.
As~\cite{YaHa06}
discuss in the context of communication networks,
an alternative justification for equilibrium is needed and is available.
In Section~\ref{Dynamics}
we consider dynamics and convergence under adaptive bid  updates
by advertisers,  and show that under smooth updating of bids,  bid
trajectories converge to the unique Nash equilibrium.

The mechanism established by our analysis is simple, flexible and
implementable. It reduces to the preferred equilibrium of the
generalized second price auction (GSP) described in~\cite{varian07, eos07} 
in the special case considered there. But GSP 
requires an ordered layout of interchangeable adverts, and 
does not readily adapt to more complex page layouts, such as text rich adverts, adverts of variable size or adverts incorporating images -- each of which are of current and increasing demand for modern online advertisement platforms. However, the flexibility to compare and price complex assignments is inherent in VCG mechanisms, and through this, we determine efficient pricing implementations for general page layouts.

It is well known in the economic literature that market clearing prices that equate to marginal utility will maximize social welfare.
However, this does not guarantee that such prices can
be implemented 
on the relevant time-scales, where adverts are assigned
per impression and charged per click, 
and search-engine-wide optimization is a highly non-trivial task. 
Our
results show that
social welfare can be optimized 
by a low complexity mechanism which assigns and prices adverts on the time-scales required for sponsored search.

\subsection{Outline}
In Section~\ref{sec:assign},  we introduce a model of sponsored search
where the platform distributes advertisers' bids over an infinitely large 
collection of keyword auctions. We define an auction mechanism where an assignment problem is solved  for each search occurrence (Section \ref{Assignment Model}). We introduce a monotonicity property, requiring  click-through rates to continuously increase with bids. 
The mechanism's pricing scheme is defined in Section \ref{Sec:pricing}.  We discuss three per-click price implementations, two are deterministic and one is randomized.

In Section~\ref{sec:opt}, we discuss the objective of platform-wide efficiency for a collection advertisers with concave utility functions.  We apply a 
decomposition argument to a social welfare optimization taken over the
uncountably infinite set of constraints in our model 
of sponsored search. The argument is based on techniques from
convex optimization and duality; proofs, complicated by the infinite setting,
are mostly relegated to the appendix. These preliminaries establish that
if 
advertisers equate bids with their marginal utility
and the platform solves a maximum weighted matching 
assignment problem for each search instance then social welfare
will be maximized. Crucially, the time-scale and information asymmetry in this
decomposition are those  relevant to sponsored search. 
The advertisers are optimizing over a slower time-scale than the platform,
and the platform uses the submitted  bids to solve   on-line
a form of generalized first price auction.

After these preliminaries, in Section~\ref{Mech} 
we make the connection with mechanism design and  
strategic advertisers. In particular, we find the form of a rebate which
incentivizes advertisers to truthfully declare bids 
that equate to their marginal utilities. This produces 
a unique Nash equilibrium which implements our decomposition, 
and this is our main result (Theorem \ref{NashTheorem}).
By the separability of our optimization, the rebate can be computed by a simple 
mechanism
requiring a single additional computation, namely the solution
 of an assignment problem, for each click-through. 
Hence  assignment and pricing occur per  search query 
and  involve straightforward
polynomial-time computations.  The platform solves a (primal)
assignment problem on a per-search time-scale; and 
advertisers maximize their payoffs by solving 
a dual optimization problem over a longer time-scale  (Proposition
\ref{riProp}).     Finally we prove the theorem, which essentially follows
from strong duality.

Section~\ref{Dynamics} contains our discussion of
dynamics and of convergence to the Nash equilibrium 
under adaptive bid  updates
by advertisers.
In Section~\ref{sec:genassignments} 
we allow more complex page layouts and                         
control of the number of positions displayed (for instance, through reserve
prices), and in 
Section~\ref{sec:ext} we allow advertisers to make
different
bids for different keywords or
categories of user. 
Section~\ref{sec:related} discusses the relationship
between our results and earlier work, and Section~\ref{sec:conclude}
concludes. 

\section{The Assignment and Pricing Model}\label{sec:assign}
We begin with 
notation that reflects a sponsored search setting, where a limited set of adverts are shown in response to 
users submitting search queries.  We let $i\in\mI$ index the finite set of 
advertisers.
Each has an advert which they wish to be shown on the pages of 
search results.
An advert, when shown, is placed in a slot $l\in\mL$. 
The set of slots is ordered, with the first (lowest ordered) slot
representing the top slot. 
Let $\tau\in\mT$ index the \textit{type} of a search conducted by a user. 
The set $\mT$ is an infinitely large set. 
The type $\tau$ may incorporate information such as the keywords, location, previous
search history, and any other information 
the platform has on the search or
searcher. As $\tau$ varies, features -- such as the keyword -- are allowed to change.
Let $p_{il}^\tau$ be
the probability of a click-through on advert $i$ if is shown in slot $l$: this
probability is estimated by
the platform and will depend on the type $\tau$.

Over time, a large number of searches from the set $\mT$ are made. We
assume these occur with distribution $\bP_\tau$. Thus we view
the click-through probability $p_{il}:\mT\rightarrow [0,1]$ as a random
variable defined on the  type space $\mT$ and with distribution 
$\bP_{\tau}$. 
For example, the random variables $p=(p_{il} : i\in\mI, l\in\mL)$
might admit a joint probability density function $f(p)$.
So, for $z=(z_{il} : i\in\mI, l\in\mL)\in [0,1]^{\mI\times \mL}$, 
\begin{equation*} 
\bP_\tau( p \leq z)= \int_{ [0,1]^{\mI\times \mL}}  \bI [p \leq z ]  f(p) {\rm d} p.  
\end{equation*}
Here 
$\bI$ is the indicator function and vector inequalities, e.g.,  $p \leq z$,
are taken componentwise, $p_{il}\leq z_{il}$ $\forall i\in\mI, l\in\mL$. 
  
We exploit the inherent randomness in $p_{il}$ 
for the optimal placement of adverts. We assume that the platform
has access to the information about the query captured in
$\tau$, and so can successfully predict the click-through probability
$p_{il}^\tau$, whilst the advertiser does not have access  to such
fine-grained search information.
Later, in Sections~\ref{Mech}
and~\ref{Dynamics}, we shall see that the platform  can use this
information asymmetry to guide the auction system towards an optimal outcome.

\subsection{Assignment Model}\label{Assignment Model}

Next we describe a mechanism by which the platform 
assigns adverts. 
Suppose advertiser $i$ submits a bid $b_i$, which reflects what the advertiser is willing to pay for a click-through.
The bid $b_i$ is a non-negative real number. Later, in Section
\ref{sec:ext}, we shall allow an advertiser to submit different bids for
different categories of search type, for example for 
different keywords.

Let $b = (b_i, i\in\mI)$.
Given the information $(\tau, b)$, the following
optimization maximizes the expected sum of bids on 
click-throughs 
from a single search.
\newline

\noindent \textbf{ASSIGNMENT}($\tau,b$)
\begin{subequations}\label{ASMT}
\begin{align}
&\rm{Maximize}  &&\sum_{i\in\mI} b_i \sum_{l\in\mL} p_{il}^\tau
x_{il}^\tau&\label{ASMT1}\\
&\rm{subject\; to} && \sum_{i\in\mI} x^\tau_{il} \leq 1, \quad  l\in\mL,
&\label{ASMT2}\\
& && \sum_{l\in\mL} x^\tau_{il} \leq 1, \quad i\in\mI,  &\label{ASMT3}\\
&\rm{over} && x^\tau_{il}\geq 0,  \quad i\in\mI, l\in\mL.\label{ASMT4} &
\end{align}
\end{subequations} 

The above optimization is an assignment problem,
where the constraint~\eqref{ASMT2} prevents a slot containing
more than one advert, and the constraint~\eqref{ASMT3} prevents 
any single advert being shown more than once on a search page.
The assignment problem is highly appealing from a computational perspective, firstly,
because an integral solution can be found efficiently (see \cite{Ku55,Be88}) and,
secondly, because there is no need to pre-compute the assignment. The
assignment problem can be solved on each occurrence of a search of type
$\tau\in\mT$, and an integral solution forms a maximum weighted matching
of advertisers $\mI$ with slots $\mL$.

We apply the convention that if $b_{i}=0$ 
then $x_{il}^\tau=0$ for $l\in\mL$,
so that a zero bid does not receive clicks. Let 
\begin{equation}
\label{ASMT9}
y_{i}^\tau = \sum_{l\in\mL} p_{il}^\tau x_{il}^\tau, \quad  \quad
 {y}_i = \bE_\tau y^\tau_{i}.
\end{equation}
The solution $x^\tau$ to the assignment problem~\eqref{ASMT}  may not be
unique: however the solution will be unique with probability one 
if, for example,
the distribution of click-through probabilities $p$ admits a
density. We make the milder assumption that $(y_i^\tau, i \in \mI)$
is unique with probability one.

Note that $y_{i}^\tau$ is the click-through rate for
advertiser $i$ from a given search page,
and ${y}_i$ is the click-through rate averaged over $\mT$.
(We shall \textit{not} use $y_{i}$ for the random variable $y_{i}^\tau$.) 
In our model the information asymmetry between the
platform and advertiser is captured by the search type $\tau$
which is known to the platform but not to the advertiser: thus
we assume that $y_{i}^\tau$ is
known to the platform, from its solution to the assignment problem,
while only the average ${y}_i$ is reported to, or accessible for estimation 
by, advertiser $i$.
For an optimal solution to the above assignment problem,
write ${y}^\tau_i  = {y}^\tau_i(b)$  to
emphasize the dependence of ${y}^\tau_i$ on the vector of bids $b$ 
and, similarly,  write ${y}_i  = {y}_i(b)$.
 Let $( b'_i , b_{-i})$ be
the vector
obtained from
$b$ by replacing the $i$th component by  $b'_i$.

We shall assume  the following 
\textit{monotonicity property} of solutions of
\text{ASSIGNMENT}($\tau,b$). 
We assume that
${y}_i(b_i, b_{-i})$ takes the value $0$ when $b_i=0$,  and is strictly increasing in $b_i$ and
continuous in $(b_i,b_{-i})$ whenever any component of $b_{-i}$ is positive.
Without the monotonicity property 
${y}_i(b)$ will be
increasing in $b_{i}$
but may not be strictly increasing or continuous.
A similar assumption has been made by \cite{Nekipelov2015EC}, who argue
that the assumption is natural and satisfied by the sponsored search data
they analyze.

The monotonicity property will generally follow from sufficient variability
of click-through rates. 
For instance, a sufficient condition is that 
the random variables $p$ admit a continuous density
$f(p)$ on the set of click-through probabilities
$\tilde{\mathcal{P}} = \{ p \in [0,1]^{|\mI|\times |\mL|} : p_{il} \geq
p_{ik}, l < k \}$ which is positive on a neighborhood containing  the origin.
Observe that on $\tilde{\mathcal{P}}$ the click-through probability
for a given advert increases as the slot it is shown in 
decreases.
The following result establishes 
the monotonicity property under the above sufficient condition.

\begin{proposition} \label{diffcont}
If the distribution $\bP_{\tau}$  admits  a
continuous probability density function on $\tilde{\mP}$ which
is positive  on a neighborhood containing  the origin
then the mapping $b_i \mapsto {y}_{i}(b_i,b_{-i})$  satisfies
the  monotonicity property.
\end{proposition}

The proof of Proposition \ref{diffcont} is given in Appendix \ref{Lemma1}. 
The sufficient condition of Proposition \ref{diffcont} is far from
necessary, as we shall illustrate later in Example \ref{example:English}. 
Our earlier assumption that $(y_i^\tau, i \in \mI)$ 
is unique with probability one is implied
by the monotonicity property.

\subsection{Pricing Model} \label{Sec:pricing}
Once adverts are allocated, prices must be determined for 
any resulting click-throughs. We consider a mechanism where the
expected rate of payment by
advertiser $i$ is

\begin{align}
\pi_i(b) {y}_{i}(b) 
&=
\int_0^{b_i} \Big({y}_i(b)-{y}_i(b'_i,b_{-i}) \Big)  db'_i .
 \label{bidpriceprimal1}
\end{align}
Here, as before, $b=(b_i: i\in\mI)$ is the vector of
advertisers' bids and $y_i(b)$ is the resulting click-through rate
for advertiser $i$. 
We discuss later the rationale for this formula
in the proper context of mechanism design, in  Section~\ref{Mech}.
For now we note that
the rate of payment~\eqref{bidpriceprimal1} can be readily
implemented
by the platform at a low computational
cost. We give three examples of implementations: the first uses
randomization to estimate the integral~\eqref{bidpriceprimal1};
the second is a form  of VCG price; and the third uses the solution 
of a linear program due
to~\cite{leonard1983elicitation}. The first two require 
the solution of 
just one additional instance of the assignment problem per click-through.

\subsubsection*{A randomized price.}
Suppose the platform solves
{ASSIGNMENT}($\tau,b$),
and observes a click-through on $(i,l)$ --- that is the solution
has $x_{il}^\tau =1$, and the user clicks on the advert in position
$l$, which is for advertiser $i$.  To price this advert, the platform
chooses
$b'_i$ uniformly and randomly on the interval $(0, b_i)$
and additionally solves {ASSIGNMENT}($\tau,(b'_i, b_{-i})$). Let
$y_i^\tau(b'_i, b_{-i}) = \sum_{l\in\mL} p_{il}^\tau x_{il}^\tau$
under a solution to this problem. The platform then charges
advertiser $i$ an amount
\begin{equation}\label{chargecalc}
b_i \left( 1 -
\frac{y_i^\tau(b'_i, b_{-i})}{ y_i^\tau(b)}     \right)
\end{equation}
for the click-through. This charge does not depend on the distribution 
$\bP_\tau$, and will lie between 0 and $b_i$. Taking expectations
over $\tau$  and $b'_i$
shows that the expected rate of payment
by advertiser $i$ is
\begin{align*}
&\bE_{\tau, b'_i} \left[
\sum_{l\in\mL} p_{il}^\tau {x}_{il}^\tau
b_i \left( 1 -
\frac{y_i^\tau(b'_i, b_{-i})}{ y_i^\tau(b)}  \right)
\right]
=   b_i  \left( {y}_i(b)  -  \bE_{b'_i}
\left[{y}_i(b'_i,
b_{-i})\right] \right)
= \int_0^{b_i} \Big({y}_i(b)-{y}_i(b'_i,b_{-i}) \Big)  db'_i, 
\end{align*}
recovering expression~\eqref{bidpriceprimal1}.

Observe that  the additional instance of the assignment problem
does not determine the assignment, and thus will not slow
down the page impression: rather, it is used to calculate
the charge~\eqref{chargecalc}
for a click-through. Indeed,  one could imagine a charge
$b_i$ on the click-through, followed by
a later rebate of a proportion $y_i^\tau(b'_i,
b_{-i})/y_i^\tau(b)$ of the charge.  
The rebate depends
on the uniform random variable $b'_i$ as well as the 
random variable $\tau$: next we shall see
that we can remove the dependence on $b'_i$.

\subsubsection*{A parametrized VCG price.}
Note that 
\begin{equation*}
\int_0^{b_i} {y}_i(b'_i,b_{-i}) db'_i =
\sum_j b_j y_j(b)
- \sum_{j \neq i} b_j y_j(0,b_{-i}),
\end{equation*}
since both expressions share the same derivative with respect to
$b_i$ (see Proposition~\ref{diffsumyprop} of Appendix~\ref{Lemma1})
and both expressions
take the value 0 when $b_i=0$. Thus the rate of payment~\eqref{bidpriceprimal1}
can be implemented by a charge
$b_i$ on a click-through followed by
a later rebate
\begin{equation}\label{chargecalc2}
\frac{1}{y_i^\tau(b)} \Big( \sum_j b_j y_j^\tau(b)
- \sum_{j \neq i} b_j y_j^\tau(0,b_{-i}) \Big).
\end{equation}
The rebate calculation again requires the solution
of one additional instance of the assignment problem, 
this time omitting advertiser $i$. 
This calculation is familiar as the VCG mechanism when
the utility function for advertiser $j$, $j \in\mI$,  is 
replaced by the surrogate linear utility $b_j y_j$. 
The charge minus the rebate has
the usual VCG interpretation as the externality caused
by advertiser $i$, but under these surrogate utilities.

\subsubsection*{Computing all prices simultaneously.} 
\cite{leonard1983elicitation} has shown that VCG prices
in assignment games are a minimal solution to a
dual assignment problem, and this allows prices for all
potential click-throughs to be calculated from the
solution to just one optimization problem. 
 
Let $A^\tau$ be the maximal value achieved by
the objective function~\eqref{ASMT1} in the assignment problem~\eqref{ASMT}.
Then  per-impression VCG prices are
given by the solution $v^\tau, s^\tau$ to the following
optimization problem.
\begin{align*}
&\rm{Minimize}  && \sum_{l\in\mL} v_l \\
&\rm{subject\; to} &&  
\sum_{i\in\mI} s_i +  \sum_{l\in\mL} v_l  =A^\tau, \\
& && s_i + v_l \geq b_i p_{il}^\tau, \qquad i\in\mI, l\in\mL, \\
&\rm{over} && s_i \geq 0, v_l\geq 0,  \qquad i\in\mI, l\in\mL. &
\end{align*}
An initial feasible solution to this dual assignment program is given by
the dual variables corresponding to an optimum of the assignment
problem~\eqref{ASMT} 
and techniques for its solution are reviewed in
\cite{BdVSV}. 

This formulation allows for either pay-per-click or pay-per-impression pricing of adverts.
After solving the problem for $v^\tau, s^\tau$, 
advertiser $i$ can either be charged
the price $v_l^\tau$ for an impression of her advert in slot $l$ or 
be charged the price $v^\tau_l/p^\tau_{il} = b_i - s^\tau_i/p^\tau_{il}$
on a click-through:
in the latter case the result
of \cite{leonard1983elicitation} 
 implies that the rebate $s^\tau_i/p^\tau_{il}$ 
will equal expression~\eqref{chargecalc2}. Observe that the dual assignment
problem to be solved is identical whichever advert is clicked on.

We end this section with a setting where 
particularly simple closed forms are available for prices.

\begin{example} \label{example:English}

If there is a single slot
then the slot will be assigned to the bidder $i$ with the
highest value of  $b_i p_{i1}^\tau$, and if
this results in
a click-through
then the charge will be $ \max_{j\neq i} b_j p_{j1}^\tau \mathbin{/} 
p_{i1}^\tau   $, 
a second price auction on the products $b_j p_{j1}^\tau$.

Suppose next there are $L$ slots with $I$  advertisers bidding and further 
suppose that the click-through probabilities take the form
$p_{il}^\tau = q_{i}^\tau p_l$ 
where 
$p_1>p_2>\ldots>p_L$. Here $p_l$ is a slot effect, and
$q_{i}^\tau$ is an advertiser effect which may depend on
the search query (for example, it may depend on some measure of distance
between the searcher and the advertiser). 
Define the search-adjusted bid $b_{i}^\tau = b_i q_{i}^\tau$ and,
given $\tau$, order the advertisers so that 
$b_1^\tau >  b_2^\tau  >  \ldots  >  b_{I}^\tau$. 
Then advertisers $1,2,\ldots, \min
\{L, I \} $ 
are allocated slots $1,2,\ldots, \min \{L, I \} $ respectively. 
If necessary ties can be broken randomly. 

In this example
it is straightforward
to calculate the expected value of expression~\eqref{chargecalc}
over $b'_i$ explicitly. Set $p_{L+1}=0$ and  $b_i = b_i^\tau = 0$ for $i> I$. Upon  
a click-through
on slot $l$  advertiser $l$ is charged the amount $\pi_l^\tau$ where
\[
\pi_l^\tau q_{l}^\tau =  b_{l+1}^\tau - \frac{1}{p_l} \sum_{m=l+1}^L p_m
(b_{m}^\tau-
b_{m+1}^\tau), \quad l = 1,2,\ldots,L. 
\]
Expressed as a recursion this implies
\begin{equation}  \label{english}
\pi_l^\tau = \frac{ q_{l+1}^\tau}{q_{l}^\tau} \left(
b_{l+1} - \frac{p_{l+1}}{p_l}(
b_{l+1} -\pi_{l+1}^\tau) \right) ,
\quad l = 1,2,\ldots,L
\end{equation}
recovering an equilibrium of the
generalized second price 
auction, \cite{eos07}. Note, however, that the
charges~\eqref{english}, and indeed the slots allocated,
fluctuate with the search type $\tau$. 
The expected revenue, given $\tau$,  is
\begin{equation}  \label{revenue}
\sum_{m=1}^L \pi_m^\tau q_m^\tau p_m
= \sum_{m=1}^L   (p_m - p_{m+1}) b_{m+1} q_{m+1}^\tau. 
\end{equation}

In the model considered by \cite{eos07} and \cite{varian07} 
the random variables $(q_i^\tau, i \in \mI)$ are all 
in fact constants, and in this case there may be multiple
Nash equilibria. For example, suppose $L=I=2$:
then for either one of the advertisers
to bid very high and the other to bid very low is a Nash equilibrium. 
We shall see in following sections that provided our monotonicity
condition is satisfied there is a unique Nash equilibrium.

The restriction that click-through probabilities 
have the product-form $p_{il}^\tau = q_{i}^\tau p_l$ 
implies they lie in a linear subspace 
of $\tilde{\mP}$: thus they do not have 
a density over $\tilde{\mP}$, and so we cannot
appeal to Proposition \ref{diffcont} to justify
the monotonicity property, in particular 
that $y_i(b)$ is a strictly increasing and continuous function of  
$b_i$. But if the advertiser effects $(q_i^\tau, i \in \mI)$
have a continuous probability density positive on 
a neighbourhood of the origin in 
$\{ q \in [0,1]^{|\mI|} \}$ then 
the monotonicity property will follow. Essentially the variability of the 
advertiser effect $q_i^\tau$ smooths out the impact of the bid $b_i$ 
sufficiently that the rate $y_i(b_i, b_{-i})$ is continuous in $b_i$.

\end{example}

%% file: DECOMPOSITION_FINAL.tex

\section{Optimization Preliminaries } \label{sec:opt}

In this section we present  an optimization problem
which 
we use to develop various decomposition and duality
results.
In particular, we find that if advertisers equate bids with their marginal utility
and the platform solves a maximum weighted matching 
assignment problem for each search instance,  then social welfare
will be maximized.

We suppose each advertiser $i$ has a utility function, $y_i \mapsto
U_i(y_i)$, where $U_i(\cdot)$ is non-negative, increasing, strictly concave
and continuously differentiable. Our objective is to place adverts 
so as to maximize the
sum of these utilities, in other words to maximize social welfare. 
To simplify the statement of results
we shall assume further that 
$U'_i(y_i) \rightarrow \infty$ as $y_i \downarrow 0$
and $U'_i(y_i) \rightarrow 0$ as $y_i \uparrow \infty$. 
The maximization of social welfare by the auction system is the
following problem. \newline

\noindent\textbf{SYSTEM}($U$, $\mI$, $\bP_\tau$)
\begin{subequations}\label{SYS}
\begin{align}
&\rm{Maximize}  &&\sum_{i\in\mI} U_i(y_i) \label{SYS1}&\\
&\rm{subject\;to} && y_{i} = \bE_\tau \Big[ \sum_{l\in\mL} p_{il}^\tau x_{il}^\tau \Big],\quad
i\in\mI, \label{SYS2} &\\
& && \sum_{i\in\mI} x^\tau_{il} \leq 1, \quad  l\in\mL, \tau\in\mT, \label{SYS3} &\\
& && \sum_{l\in\mL} x^\tau_{il} \leq 1, \quad i\in\mI, \tau\in\mT, \label{SYS4} &\\
&\rm{over} && x^\tau_{il}\geq 0, y_i\geq 0 \quad i\in\mI, l\in\mL.\label{SYS5} & 
\end{align}
\end{subequations}
Inequalities~\eqref{SYS3} and~\eqref{SYS4} are just the
scheduling 
constraints~\eqref{ASMT2} and~\eqref{ASMT3}, that each slot can show at most
one advert and that each advertiser can show at most one
advert, while equality~\eqref{SYS2} recaps the definition~\eqref{ASMT9}
of $y_i$, the expected click-through rate. Over these constraints we maximize social welfare, i.e., the aggregate sum of the utilities.

To solve the above optimization, one could imagine that there is a
centralized designer who knows everything about the entire
system: the advertisers' utilities $U_i(\cdot), i\in\mI$, click-through
probabilities  $p_{il}^\tau, i\in\mI, l\in\mL,
\tau\in\mT$, and the  distribution $\bP_\tau$ over these probabilities. 
This designer then attempts to assign adverts in a way so that $y_i,
i\in\mI$, the
click-through rates received by advertisers, maximize social welfare.
The solution of such an optimization by centralized means is not
possible --- for example, the utilities will not be known --- 
but the form of the solution will help us
develop an appropriate decomposition, respecting the 
time-scales relevant to the platform and advertisers. 
In the next section, on mechanism design, we
consider the game theoretic aspects that arise when, instead of
a single system optimizer, the platform
and advertisers have differing information and incentives. 

Incorporating the constraint~\eqref{SYS2} into the objective 
function~\eqref{SYS1} gives the Lagrangian 
\begin{align*} 
L_{sys}(x,y;b)&=\sum_{i\in\mI} U_i(y_i) +\sum_{i\in\mI} b_i
\bE_\tau\! \left[ \sum_{l\in\mL} p_{il}^\tau x_{il}^\tau - y_{i}
\right],
\end{align*}
where $b_i, i \in \mI$ are the Lagrange multipliers associated with the constraints ~\eqref{SYS2}, with $b_i \geq 0$. 
Notice, we intentionally omit the scheduling constraints from 
our Lagrangian. Thus we seek to maximize the Lagrangian
subject to the constraints~(\ref{SYS3}-\ref{SYS4}) as well
as~\eqref{SYS5}. Let $\mS$ be the set of variables 
$x^\tau=(x^\tau_{il} : i\in\mI, l\in\mL)$ satisfying the
assignment constraints (\ref{ASMT2}-\ref{ASMT4}), and let 
$\mA$ be the set of
variables $x=(x^\tau \in\mS: \tau \in \mT)$
satisfying the assignment constraints~(\ref{SYS3}-\ref{SYS5}).
We see that our Lagrangian problem
is separable in the following sense
\begin{subequations}\label{L1sep}
\begin{align}
\max_{x\in\mA, y \geq 0} L_{sys}(x,y;b)
&=\sum_{i\in\mI} \max_{y_i\geq 0}\left\{ U_i(y_i)-b_iy_i \right\}
\label{L1sep1}\\
&+  \bE_\tau \left[ \max_{x^\tau \in \mS} \sum_{i\in\mI}\sum_{l\in\mL}
b_i p_{il}^\tau x_{il}^\tau \right] \label{L1sep2}.
\end{align}
\end{subequations}

Define 
\begin{equation}\label{USER}
 U^*_i(b_i)=\max_{y_i\geq 0}\left\{ U_i(y_i)-b_iy_i \right\}. 
\end{equation}
The optimization over $y_i$ contained in the definition~\eqref{USER} 
would arise if advertiser $i$ were presented with a fixed price
per click-through of $b_i$: if allowed to choose freely her click-through rate, 
she would then choose $y_i$ such that $U'_i(y_i)=b_i$.
By our assumptions on $U_i(\cdot)$, this equation has a unique solution
for all $b_i \in (0, \infty)$.
Call $ D_i(\xi) =  \{U'_i\}^{-1}(\xi)$ the \textit{demand}
of advertiser
$i$ at price $\xi$. It follows that $U^*_i(b_i)$ can
be written in the form
\begin{equation}\label{DU}
U^*_i(b_i) = \int_{b_i}^\infty D_i(\xi) d\xi;
\end{equation}
call this advertiser $i$'s \textit{consumer surplus} at
the price $b_i$.
From this expression we can deduce that $U^*_i(b_i)$  is positive, decreasing, strictly convex and continuously differentiable.

Observe that the maximization inside the
expectation~\eqref{L1sep2}
is simply the problem {ASSIGNMENT}($\tau, b$), and thus 
we can write
\begin{equation*}
\max_{x\in\mA, y \geq 0} L_{sys}(x,y;b)
=\sum_{i\in\mI}   U^*_i(b_i) +  \sum_{i\in\mI} b_i y_i(b)
. 
\end{equation*}
The Lagrangian dual of the SYSTEM problem~\eqref{SYS} can thus be written
as follows.  \newline

\noindent\textbf{DUAL}($U^*$,$y$, $\mI$)
\begin{subequations} \label{SYSdual}
\begin{align} 
&\rm{Minimize}  &&\sum_{i\in\mI}  \left( U^*_i(b_i) + 
b_i y_i(b) \right) &\\
&\rm{over} && b_i \geq 0,  \quad i\in\mI. &
\end{align}
\end{subequations}
Owing to the size of the type space $\mT$, the
optimization~\eqref{SYS} has a potentially uncountable number of
constraints. This presents certain technical difficulties, for instance
those associated with  proving strong duality. These issues are dealt with
in the appendix, where the proofs of the following two propositions  are presented. 

We first observe that  the SYSTEM problem decomposes into optimizations
relevant to the advertisers and to the platform.

\begin{proposition}[Decomposition]\label{Decomp1}
Variables $\tilde{y}$, $\tilde{x}^\tau, \tau\in\mT$, satisfying the
feasibility conditions~(\ref{SYS2}-\ref{SYS5})
are optimal for
{SYSTEM}($U$,$\mI$,$\bP_\tau$) if and only if there exist $\tilde{b}_i$,
$i\in\mI$, such that
 \begin{enumerate}[A.]
\item $\tilde{b}_i$ minimizes $U^*_i(b_i) +  b_i
\tilde{y}_i$ over $b_i \geq 0$, for each $i \in \mI$, 
\item $\tilde{x}^\tau$ solves {ASSIGNMENT}($\tau, \tilde{b}$) with
probability one under the distribution $\bP_\tau$ over $\tau\in\mT$. 
\end{enumerate}
\end{proposition} 

In this proposition, the optimization in Condition A does not naturally correspond to the bidding behavior of strategic advertisers, at least in its present form.   
Hence we need to examine the  implications of Condition A for the
construction of prices \eqref{bidpriceprimal1} that do  give strategic advertisers the incentive to solve the SYSTEM problem. We do this in 
the next section, Section \ref{Mech}. There we shall also see that 
the per-click pricing implementations \eqref{chargecalc} and
\eqref{chargecalc2} are made possible by the decomposition
into per-impression assignments, Condition  B.

The optimal bids $\tilde{b}$
can be further understood through the following dual characterization.

\begin{proposition}[Dual Optimality]\label{NashProp}$\;$\\
a)  The objective of the dual problem~\eqref{SYSdual} is continuously
differentiable for $b>0$ and is minimized uniquely by the 
positive vector $\tilde{b}=(\tilde{b}_i: i\in\mI)$ satisfying, for each $i\in\mI$,
\begin{equation}\label{DualOptCond}
\frac{d U^*_i}{d b_i}(\tilde{b}_i) + {y}_i(\tilde{b})=0. 
\end{equation}
\noindent b) If $\tilde{b}$ is an optimal solution to the DUAL problem~\eqref{SYSdual}  then ${x}^{\tau}(\tilde{b})$, ${y}(\tilde{b})$ are optimal for the SYSTEM problem \eqref{SYS}.
\end{proposition}
The dual provides a finite parameter optimization from which the SYSTEM
problem can be solved. Moreover, \eqref{DualOptCond} provides conditions on
advertiser demands which, to solve the SYSTEM problem, must be effected by the auction system in strategic form.

%% file: MECH_DESIGN_FINAL.tex

\section{Mechanism Design} \label{Mech}
We now prove that our mechanism implements our system optimization.
In the last section we demonstrated how this global
problem can be decomposed into two types of sub-problem: one, where the
platform finds an optimal assignment given click-through probabilities; 
and the other, where the dual variables $b$ are each set to 
 solve a certain single parameter dual problem. 
In this section we suppose the advertisers act strategically, anticipating
the result of the platform's assignment and attempting
to maximize their payoff.

Henceforth $b_i$ is the
\textit{bid}
submitted by advertiser $i$ and, as a function of these bids, we formulate
prices 
that incentivize the advertisers to choose bids that 
result in an assignment that solves 
the SYSTEM problem~\eqref{SYS}.

Consider a mechanism where, given the vector of bids
$b=(b_i:i\in\mI)$, each advertiser, $i$,  receives a
click-through rate ${y}_i(b)$, and from this derives a benefit
$U_i({y}_i(b))$ and is charged an expected price $\pi_i(b)$
per click. The payoff to advertiser $i$ arising from 
a vector of bids 
$b=(b_i : i\in\mI)$ is then 
\begin{equation}\label{rewardsprimal}
{u}_i(b) = U_i ({y}_i(b))  - \pi_i(b) {y}_i(b).
\end{equation}
A 
\textit{Nash equilibrium} is a vector of bids 
$b^*=(b^*_i : i\in\mI)$ such that, for $i\in\mI$  and all
$b_i$
\begin{equation}\label{Nashr}
{u}_i(b^*) \geq {u}_i ( b_i , b_{-i}^*). 
\end{equation}
Here $( b_i , b_{-i}^*)$ is  obtained from 
 the vector $b^*$ by replacing the $i$th component by  $b_i$.

The main result of this section is the following.

\begin{theorem}\label{NashTheorem}
If prices are charged so that the expected
rate of payment by advertiser $i$, for $i \in \mI$, is given
by expression~\eqref{bidpriceprimal1} 
then there exists a unique Nash equilibrium, and it is 
given by the vector of optimal prices identified in 
Proposition~\ref{NashProp}. Thus the assignments achieved at the Nash
equilibrium, $x^\tau(b^*), y(b^*)$, form a solution to the 
SYSTEM problem~\eqref{SYS}. 
\end{theorem}

The result states that, given adverts are assigned according to the
assignment problem \eqref{ASMT}, the game theoretic equilibrium reached by
advertisers attempting to maximize their respective payoffs ${u}_i$
solves the problem {SYSTEM}($U$,$\mI$,$\bP_\tau$). 
Since ${y}_i(b'_i,b_{-i})$ is a strictly increasing
function of the bid $b'_i$, it follows from~\eqref{bidpriceprimal1} that  the price $\pi_i(b)$ must be strictly
lower than the bid $b_i$. Setting a price lower than the submitted
bid is a prevalent feature of online auctions used by search engines, and,
as we emphasized in Section \ref{sec:assign}, the prices
\eqref{bidpriceprimal1} can be practically implemented in a sponsored search setting.

{We note that, in this section, each advertiser expresses their preferences through a single bid. This framework extends naturally to the case where advertisers place multiple bids over multiple different keywords (or search categories). This extension is given in Section \ref{sec:ext}.}

\subsection{Proof of Theorem \ref{NashTheorem}}
\label{sec:proof-Nash}

To establish Theorem~\ref{NashTheorem} 
we will require 
an additional result, Proposition \ref{riProp}, which indicates
how 
maximal {payoffs } achieved by each advertiser relate to the solution of the
dual problem, given by Proposition \ref{NashProp} from the previous
section.

\begin{proposition}[Mechanism Dual]\label{riProp} For each positive choice
of $b_{-i}=(b_{j} : j\neq i, j\in\mI)$, the following equality holds
\begin{align}
&\max_{b_i\geq 0} {u}_i(b) = \min_{b_i\geq 0} 
	\left\{ 
		U^*_i(b_i) + \int_0^{b_i} {y_i}(b'_i,b_{-i}) db'_i
	\right\}. \label{rewards4}
\end{align}
Moreover, the optimizing $b_i$ for both expressions is the same,
is unique and finite, and satisfies 
\begin{equation} \label{DualOptCond2}
\frac{d }{d b_i} U^*_i(b_i) + {y}_i(b) = 0.
\end{equation}
\end{proposition}

\begin{proof}
We calculate the conjugate dual of the {payoff } function \eqref{rewardsprimal}. Let
$P_i(y_i)$ be the function whose Legendre-Fenchel transform is 
\begin{equation*}
P_i^*(b_i)=\int_0^{b_i} {y_i}(b'_i,b_{-i}) db'_i.
\end{equation*}
The above function is increasing and convex, and we know from Fenchel\rq{}s Duality Theorem \cite[Theorem 3.3.5]{borwein2006convex} that
\begin{equation} \label{FDT2}
\max_{y_i\geq 0}
	\left\{ 
		U_i(y_i) - P_i(y_i) 
	\right\}
=
\min_{b_i \geq 0} 
	\left\{ 
		U^*_i(b_i) + P_i^*(b_i) 
	\right\}.
\end{equation}
Next we 
calculate the function $P_i$ from the dual of the function $P^*_i$ above.
By the Fenchel--Moreau Theorem, \cite{borwein2006convex},  we know this to be
\begin{equation*} 
P_i(y_i)
 =
\min_{b_i\geq 0} 
	\left\{ 	
		b_i y_i  -\int_0^{b_i} 	{y_i}(b'_i,b_{-i}) db'_i 
	\right\}.
\end{equation*}
The optimum in this expression occurs when
${y}_i(b) = y_i$.
Substituting this back, since $b_i\mapsto {y}_i(b)$ is strictly increasing, we have that
\begin{align}\label{Peqn}
P_i(y_i) 
& =
	\int_0^{\infty} \left( y_i - {y}_i(b'_i,b_{-i})\right) \bI [ {y}_i(b'_i,b_{-i}) \leq  y_i ] db'_i.
\end{align}
In other words, as expected with the Legendre-Fenchel transform, the area under the curve ${y}_i(b_i, b_{-i})$ is converted to the area to the left of the curve ${y}_i(b_i, b_{-i})$.
Further, notice, if $y_i > \max_{b_i} {y}_i(b_i,b_{-i})$ then
$P_i(y_i) = \infty$, 
and thus the finite range of the function $y_i\mapsto P_i(y_i)$ is exactly
the same as that of $b_i \mapsto  P_i(y_i(b))$.
Noting \eqref{Peqn} and this last observation, the equality \eqref{FDT2} now reads
\begin{align*}
\min_{b_i\geq 0} 
	\left\{ 
		U^*_i(b_i) + \int_0^{b_i} {y}(b'_i,b_{-i}) db'_i
	\right\} 
 = &\max_{y_i \geq 0 } \left\{ U_i(y_i) -P_i(y_i)  \right\} \\
 = & \max_{b_i \geq 0} \left\{ U_i({y}_i(b)) -P_i({y}_i(b))  \right\} \\
= & \max_{b_i \geq 0} \left\{ U_i({y}_i(b)) -\pi_i(b) {y}_i(b)  \right\}. 
\end{align*}
In the final equality we note from the definition~\eqref{bidpriceprimal1}
that 
$P_i({y}_i(b)) = \pi_i(b) {y}_i(b)$. 
This gives the equality \eqref{rewards4}.

We now show that both expressions \eqref{rewards4} are determined at the
same unique value of $b_i$. The function $U^*_i(b_i)-P^*_i(b)$ is a
strictly convex differentiable function of $b_i$, whose unique minimum is
given by the required expression~\eqref{DualOptCond2}. 
Further, $b_i\mapsto y_i(b_i,b_{-i})$ is strictly increasing and $b_i$ achieves the range of the strictly concave function $U_i(y_i)-P_i(y_i)$ under $y_i=y_i(b_i,b_{-i})$.  Thus  $U_i(y_i)-P_i(y_i)$ is maximized uniquely by $y_i=y_i(b_i,b_{-i})$ (and thus uniquely by $b_i$) satisfying
\begin{equation}\label{Ucondy2}
\frac{d}{dy_i} U_i(y_i(b)) - b_i = 0.
\end{equation}
Since $\frac{d}{db_i}U^*_i$ is the inverse of the strictly increasing
function $\frac{d}{dy_i}U_i$, it is clear 
that \eqref{DualOptCond2} and \eqref{Ucondy2} are equivalent and satisfied by the same unique $b_i$. This completes the proof.
  \end{proof}

The proof of Theorem \ref{NashTheorem} follows by observing the optimality conditions of Propositions \ref{NashProp} and \ref{riProp}.

\begin{proof}[Proof of Theorem \ref{NashTheorem}]
Before proceeding with the main argument, we note that a Nash equilibrium
must be achieved by positive values of $b_i$. By applying the mean
value theorem, for some $\tilde{y}$ satisfying $0=y_i(0,b_{-i}) \leq
\tilde{y} \leq y_i(b_i,b_{-i})$, we have
\begin{align}
{u}_i(b_i,b_{-i}) &\geq U_i (0) + U_i'(\tilde{y}) (y_i(b_i,b_{-i}) - y_i(0,b_{-i}) ) - \int_0^{b_i} (y_i(b_i,b_{-i}) - y_i(0,b_{-i})) db'_i \notag \\
&= {u}_i(0,b_{-i}) + (U'_i(\tilde{y})-b_i )( y_i(b_i,b_{-i}) - y_i(0,b_{-i})) \label{AB}\\
& > {u}_i(0,b_{-i}). \notag
\end{align}
The second term in \eqref{AB} is positive for $b_i$ sufficiently small, since $U'_i(\tilde{y})-b_i\nearrow \infty$ as $b_i\searrow 0$ and from our monotonicity property $y_i(b_i,b_{-i}) > y_i(0,b_{-i})$. From this we see that a Nash equilibrium can only be achieved with $b_i>0$ for each $i\in\mI$.

By  Proposition \ref{riProp}, $b=(b_i:i\in\mI)>0$ 
is a Nash equilibrium 
if and only if condition~\eqref{DualOptCond2} is satisfied
for each $i\in\mI$. But by Proposition \ref{NashProp}b), these
conditions hold if and only if 
$b$ is the unique solution to the dual to the SYSTEM problem. So, the set of
Nash equilibria are the optimal prices defined for the decomposition,
Proposition~\ref{Decomp1}. By Proposition~\ref{NashProp}b), the assignment
achieved by Nash equilibrium bids maximizes the utilitarian objective
{SYSTEM}($U$, $\mI$, $\bP_\tau$). Finally, by Strong Duality  (Theorem
\ref{AppendixTheorem} of Appendix~\ref{sec:Lagrange}), there exists $b^*$ which optimizes the dual
problem \eqref{SYSdual}, and thus there must be a Nash equilibrium.
  \end{proof}

\begin{remark}
The optimality condition \eqref{DualOptCond} or  \eqref{DualOptCond2} 
states that each advertiser\rq{}s demand, $D_i(b_i)$,
and supply, ${y}_i(b)$, should equate, and is a consequence
of the Envelope Theorem. 
A more familiar context for this form of result is 
Vickrey pricing~(\cite{Vick61}) and  
Myerson\rq{}s Lemma (or the Revenue Equivalence Theorem), see
\cite{myerson1981optimal} and \cite[Theorem 3.3]{milgrom2004putting},
which are also  consequences of the Envelope Theorem. 
But observe that we are using general utilities, which despite the single input parameter $b_i$, takes us out of a single parameter type space to which Myerson's Lemma generally applies.

\end{remark}

We have assumed throughout
the monotonicity property, ensuring that the mapping $b_i\mapsto {y}_i(b_i, b_{-i})$
is strictly increasing and continuous.
A natural question concerns
whether the monotonicity property can be relaxed. 

\begin{example}\label{FrankExample}
If the mapping is discontinuous, there may be inefficient Nash
equilibria, and the $L=I=2$ case discussed
in Example~\ref{example:English}, with two advertisers
and two slots, provides an illustration. 
The same difficulty can arise even if the mapping is continuous
but not strictly increasing, as we now show. 
Amend the illustration, by supposing that 
the advertiser effects $q_1^\tau, q_2^\tau$ are independent
random variables with continuous probability density functions
each supported on the interval $(q-\epsilon, q + \epsilon)$ for
$q >> \epsilon >0$.  The mapping $b_i\mapsto {y}_i(b_i, b_{-i})$ is now
continuous, although not strictly increasing.
The inefficient Nash equilibria remain,
where one of the advertisers
bids  very high and the other very low. 
If we assume the densities of $q_1^\tau, q_2^\tau$ are 
positive in a neighbourhood of the origin, then the 
mapping is necessarily strictly increasing, because 
a small increase in an advertiser's bid will have a small
but positive probability of improving the slot allocated
to the advertiser:
competition exists between the advertisers, whatever
their bids, for at least some search types $\tau$, and this
ensures the uniqueness and efficiency of the Nash equilibrium. 
\end{example}

%% file: DYNAMICS_FINAL.tex

\section{Dynamics and Convergence} \label{Dynamics} 

We have seen in Section~\ref{sec:assign} that our assignment model involves
the rapid solution of a computationally straightforward problem
for each individual search. The challenge facing an advertiser is of
a different form: she has to rely on noisy and possibly delayed
feedback averaged over some period of time in order to learn
the mean click-through rate $y_i$ that has been achieved by
her bid $b_i$, and she then has to decide whether to 
vary her bid. We shall formulate the advertiser's
problem in continuous time,
and the natural question is whether multiple 
advertisers smoothly varying their bids $b_i(t)$ as a consequence
of their current click-through rates $y_i(t)$ will
converge to the Nash equilibrium.

Convergence may not be possible when the search space is discrete, e.g., for an auction on a single search type. 
Essentially, the search engine does not have enough additional
information from the search type $\tau$ to fine tune its discrimination
between advertisers. 
However, in sponsored search, there is inherent variability in the search
type $\tau$ which
will influence the click-through probabilities of the advertiser.
This is the motivation for our assumption of the monotonicity property, 
that the
distribution $\bP_{\tau}$ over $\mT$ 
is such that the click-through rate 
${y}_i(b)$ is a continuous, strictly increasing function of
$b_i$. We shall see that, 
under models of advertiser response, we are then 
able to deduce convergence towards a system optimum.

Recall the objective function for the dual of the system problem as derived in Proposition \ref{NashProp},
\begin{equation} \label{eq:Lyapunov}
\mathcal V(b)  = \sum_{i\in\mI} U_i^*(b_i) +  
\sum_{i\in\mI}  b_i
y_i(b)   . 
\end{equation}
This expression is the
sum of the consumer surpluses and the revenue
achieved by the platform at prices $b$ and,
when $b$ is optimal, it is equal to
the maximal total welfare as defined by the SYSTEM problem \eqref{SYS}.
Further,  $\mathcal V(b)$ is continuously differentiable for $b > 0$ with 
\begin{equation*}
\frac{\partial \mathcal V}{\partial b_i}=
-D_i(b_i) + y_i(b).
\end{equation*}

We next model advertisers' responses to their observation of
click-through rates. 
We suppose advertiser $i$ changes her bid $b_i(t)$ smoothly
(i.e., continuously and differentiably)
as a consequence of her observation of her 
current click-through rate $y_i(t)$
so that
\begin{equation}  \label{eq:dynamics} 
\frac{d}{dt}  b_i(t) \gtrless 0 \text{ according as }
b_i(t) \lessgtr U_i'(y_i(b(t))). 
\end{equation}
This is a natural dynamical system representation of advertiser $i$ varying
$b_i$ smoothly in order to 
improve her payoff ${u}_i(b)$, given by
expression~\eqref{rewardsprimal}, 
 under
prices~\eqref{bidpriceprimal1}, since under the monotonicity
condition a small positive change in $b_i$ will cause a small positive change
in $y_i(b)$ and the impact on ${u}_i(b)$ will be positive
or negative as in relation~\eqref{eq:dynamics} - see Lemma~\ref{Lemma:dyn}
in
Appendix~\ref{sec:dyn}. Note that from
the definition of the demand function $D_i(\cdot)$,
\begin{equation}  \label{eq:equivdynamics}
y_i \lessgtr D_i(b_i)   \text{ according as }
b_i \lessgtr U_i'(y_i).
\end{equation}
The payoff 
${u}_i(b)$ is maximized over $b_i$ when $b_i$ and $U_i'(y_i(b))$ equate, or
equivalently, when $y_i(b)$ and $D_i(b_i)$ equate.

\begin{theorem}[Convergence of Dynamics] \label{CONV}
Starting from any point $b(0)$ in the interior of the positive orthant,
the trajectory $(b(t): t \geq 0)$ of the above dynamical system
converges to a solution of the {DUAL} problem~\eqref{SYSdual}. Thus
${y}(b(t))$, the assignment achieved by the prices $b(t)$,
converges to a solution of the {SYSTEM} problem~\eqref{SYS}.
\end{theorem}
\begin{proof}
We prove that the objective of the dual problem $\mV(b)$, defined
above,  is a Lyapunov function for the dynamical system. 
Note that $\mV(b)$ is continuously differentiable for $b>0$.
Since $y_i(b) \downarrow 0$ as $b_i \downarrow 0$ and $U'(0)>0$ 
it follows from~\eqref{eq:dynamics} that
there exists $\delta >0$, possibly depending on $b_{-i}(t)$, such that
$\frac{d }{dt}b_i(t) > 0$ if $b_i(t)\leq \delta$. 
We deduce that the paths of our dynamical system $(b(t):
t \geq 0)$ are strictly positive and $\mV(b(t))$ is continuously 
differentiable along these paths. Further, the level sets $\{ b : \mV
(b) \leq \kappa\}$ are compact: this is an immediate consequence of
the facts that the functions $U_i^*(b_i)$ are positive and decreasing, 
and, as proven in Lemma~\ref{diffsumyprop}, that
\begin{equation*}
\lim_{||b || \rightarrow\infty} \sum_{i\in\mI} b_i{y}_i(b) = \infty.
\end{equation*}

Differentiating $\mV(b(t))$ yields
\begin{align*}
\frac{d}{dt}\mathcal V(b(t))
&= \sum_{i\in\mI} \frac{\partial \mathcal V}{\partial b_i}
\frac{d}{dt} b_i(t) = -\sum_{i\in\mI} \left( D_i(b_i(t)) - y_i(b(t)) \right)
\frac{d}{dt} b_i(t) \leq 0,
\end{align*}
where the inequality follows from relations~\eqref{eq:dynamics}
and~\eqref{eq:equivdynamics}, and is strict unless 
$D_i(b_i(t)) = {y}_i(b(t))$ for $i\in\mI$. 
By Lyapunov\rq{}s Stability Theorem, see \cite[Theorem
4.1]{khalil2002nonlinear}, the process $(b(t) : t\geq 0)$ converges to the
set of points $b^*$ satisfying, for $i\in\mI$, $D_i(b_i) =
{y}_i(b)$. 
Recall that $\tfrac{dU^*_i}{d b_i}=-D_i(b_i)$ and thus, 
by Proposition \ref{NashProp}(a), the price process $b(t)$
converges to an optimal solution to the dual problem \eqref{SYSdual}. 
By the monotonicity property ${y}(b)$ is continuous, and
thus by Proposition \ref{NashProp}(b) 
the click-through rates
${y}(b(t))$ converge to an optimal solution for the system problem.
$\,$\end{proof}

In the above discussion we model advertisers that smoothly change their 
bids over time. However, we remark that other convergence mechanisms could
be considered. For instance, since our dual optimization problem 
is convex and continuously differentiable, we can minimize the dual through a coordinate descent algorithm,
where each component $b_i$ is sequentially minimized. Such an
algorithm could correspond to a game played sequentially with advertisers iteratively maximizing over
$b_i$ their payoff 
${u}_i(b_i,b_{-i})$. Previous work on global convergence to a Nash
equilibrium using an assumption of local rather
than complete information is described by~\cite{YaHa06}.

The dynamical system of this section
allows advertisers' bids to adapt to a non-stationary environment,
for example if the set of participating advertisers changes. 
Note that we have left unexplored the statistical aspects 
of estimating the click-through rates $y_i$,
although some insights
are available from the 
network utility maximization framework for communication
networks (\cite{kelly2003}). In particular, if the period
of time over which the click-through rates $y_i$ are
estimated is longer then this will improve the statistical accuracy 
of the estimation, but will also slow down 
the rate of adaptation to a changing environment; and 
even in a stationary environment there 
is necessarily a trade-off between the speed of
convergence to, and the stochastic variability around, the system optimum.

\begin{remark}
We have assumed
$U_i(\cdot)$ is non-negative, increasing and strictly concave, and is
continuously differentiable with boundary conditions 
$U'_i(y_i) \rightarrow \infty$ as $y_i
\downarrow 0$ and $U'_i(y_i) \rightarrow 0$ as $y_i \uparrow \infty$.
The boundary conditions have 
simplified the statement of results, but
are not critical. 
If we assume only that
$U_i(\cdot)$ is increasing, strictly concave, and
continuously differentiable with $U'_i(0) < U'_i(\infty)$ then
the Lyapunov function~\eqref{eq:Lyapunov} remains strictly
convex on the domain $\{ b:  U'_i(\infty)< b_i < U'_i(0), i \in \mI \}$,
has an interior minimum, and starting from any point $b(0)$
in this domain the trajectory $(b(t): t \geq 0)$
converges to the point $b$ achieving this minimum, which is 
the unique Nash equilibrium.
\end{remark}

%% file: EXTENSIONS_FINAL.tex

\section{General Assignments} \label{sec:genassignments}

In this section we consider how the
assignment problem~\eqref{ASMT} can be generalized
within our framework. Some extensions are immediate
and straightforward. For example, we could allow 
the number of slots $L = L(\tau)$ to depend on the search type $\tau$;
the pricing implementations of Section~\ref{Sec:pricing}
do not require $L(\tau)$ to be constant over $\tau$.
In this Section we consider two further
generalizations of practical importance.

\subsection{More complex page layouts} \label{complexlayouts}

Suppose
the platform allows adverts of different sizes:
for example, an advertiser may wish to offer an
advert that occupies two adjacent slots. More generally
adverts may vary in size, position, and include images and other media.
So the platform
may have a more complex set of possible page layouts
than simply an ordered list of slots $1,2,\ldots,L$.
Let $l \in \mathcal{L}$ describe a possible layout of
the adverts for advertisers $i\in\mI$. Let $p_{il}^\tau$ be
the probability of a click-through to advertiser $i$
under layout $l$. Then the
generalization of the assignment problem~\eqref{ASMT} becomes
\begin{align*}
&\rm{Maximize}  &&\sum_{i\in\mI} b_i p_{il}^\tau &\\
&\rm{over} && l \in \mathcal{L}.&
\end{align*}
Indeed, this formulation allows the click-through probabilities for an
advert
to depend not just on the advertiser and the position within the page,
but also on which other adverts are shown on the page, provided
only  the probabilities $p_{il}^\tau$ can be estimated.

The complexity of this optimization problem depends on the
design of the page layout through the
structure of the set $\mathcal{L}$ and may depend on any structural
information on the probabilities $p_{il}^\tau$,
but for a variety of cases  it will remain an assignment problem
with an efficient solution.
If  $y_i(b)$ is again
defined as the expected click-through rate for advertiser $i$ from
a bid vector $b$, and if it satisfies the monotonicity property, then
 Theorems~\ref{NashTheorem}
and~\ref{CONV} hold with identical proofs.

\begin{example}
In an image-text auction, the platform may place on a page
\emph{either} an ordered set of text adverts
(as described in Section~\ref{sec:assign})
\emph{or} a single image advert. As before advertiser $i$ bids $b_i$, the
marginal utility to advertiser $i$ of an additional click-through;
and now we suppose advertisers $i \in \mI_{text}$  make available
 text adverts and
advertisers $i \in \mI_{image}$ make available image adverts, where
$\mI = \mI_{text} \cup \mI_{image}$ and an advertiser $i \in \mI_{text}
\cap
\mI_{image}$ makes available both a text and an image advert. Let the
click-through probability on image advert $i$
be $p_{i}^\tau$ for $i \in \mI_{image}$,
with click-through probabilities on text adverts as in
Section~\ref{sec:assign}.

For this example the assignment problem is straightforward: the platform
solves the earlier assignment problem~\eqref{ASMT}
over advertisers $i \in \mI_{text}$, and shows text adverts
if the optimum achieved exceeds
$\max_{ i \in \mI_{image}} b_i p_{i}^\tau  $
and otherwise shows an image achieving this latter maximum. Similarly
the calculation of the rebate is straightforward, with one further
assignment problem to be solved for each click-through.

It is of course possible to construct assignment problems that
are not as straightforward.
For example, suppose that adverts are of different sizes,
and the platform has a bound on the sum of the advert sizes shown. The
assignment problem then includes as a special case the knapsack problem.
In general the problem is NP-hard but it becomes computationally
feasible if, for example, there are a limited number of possible
advert sizes, as in the image-text auction above.
\end{example}

\subsection{Controlling the number of slots} \label{sec:controlslots}

The platform may wish to limit the number of slots filled,
if it judges the available adverts as not sufficiently interesting
to searchers. Ultimately showing the wrong or poor quality adverts 
can cause searchers to move platform and so hurt long-term platform
revenue.

Suppose the platform
judges there is a benefit (positive or negative) $q_{il}^\tau$
to a searcher
for an impression of the advert from advertiser $i$ in slot $l$
for a search of type $\tau$, regardless of  whether
or not the searcher clicks on the advert.
The system objective function~\eqref{SYS1} then  becomes
\[
\sum_{i\in\mI} U_i(y_i) +  \bE_\tau \left[ \sum_{i\in\mI}\sum_{l\in\mL}
q_{il}^\tau x_{il}^\tau  \right],
\]
the assignment objective function~\eqref{ASMT1} becomes
\[
\sum_{i\in\mI} \sum_{l\in\mL} \left( b_i  p_{il}^\tau
 + q_{il}^\tau  \right) x_{il}^\tau,
\]
and our results hold with minor amendments.
In particular,
equation~\eqref{bidpriceprimal1} for the price function and
equation~\eqref{eq:Lyapunov} for the Lyapunov function are unaltered,
although of course the functions $y_i(b)$ will
now be defined in terms of solutions to the new assignment problem.

An important special case occurs when $q_{il}^\tau \equiv - R$, where $R$
is
a \emph{reserve price}, but in this case we need to slightly
perturb the set-up to ensure that
Proposition \ref{diffcont} remains sufficient for the monotonicity
property.
Suppose that
$q_{il}^\tau = q^\tau$ for all $i \in \mI, l \in \mL$
where $q^\tau =0$ or  $-R$  with probabilities
$\epsilon$ and $1- \epsilon$ respectively. (Formally, augment the space
$\mT$ to carry a random variable $q^\tau$ that is independent of
the click-through probabilities $p_{il}^\tau$.) Then with probability
$1-\epsilon$ an
advert will be shown in a slot only if its contribution to
 the objective function  of the assignment problem,
$b_i p_{il}^\tau$, is at least $R$. With probability
$\epsilon$  a reserve is not applied: we add the possibility to
ensure $y_i(b)$ is increasing even for small $b_i$.

Of course a reserve $R$ may also have a favourable effect
on the revenue received by the platform,
\cite{OstrovskySchwarz2011,BCKKK:14}.
As an illustration, consider
the generalized second price auction of Example~\ref{example:English}.
A reserve of $R$ will reduce the number of slots filled
if $R > b_{L} p_L^\tau $ and may increase
the revenue, which can be calculated from expression~\eqref{revenue}.
Nevertheless our framework is one of utility maximization: we
assume the platform is trying to assure its long-term revenue
by producing as much benefit as possible for its users, its
advertisers and itself.  There are, of course, several ways
in which the platform could
increase its own revenue within the utility maximization framework:
in the absence of competition from other platforms
it could for example charge an advertiser a fixed fee, less
than the advertiser's consumer surplus, to  participate.

As yet a further example of the flexibility of the framework,
instead of a fixed reserve price
we could allow an organic search result $k$ to compete for a slot,
with a positive benefit $q_{kl}^\tau$, but with $b_k = 0 $.
Recent work has analyzed
the trade-off in objectives between the platform
and advertiser in sponsored
search:  \cite{roberts13}  focus on ranking algorithms, trading off revenue
against welfare, while \cite{BCKKK:14} also include the user as an
additional stakeholder.   Our framework aims to maximize the
aggregate social
welfare of the auction system, but it is noteworthy that
this simple model of the benefit to a user of organic search
results can be subsumed within our framework.

\section{Platform-wide optimization} \label{sec:ext} 

An advertiser may judge some types of click-through as 
more valuable than others. 
In this section we suppose that the platform allows
an advertiser to express such preferences, by 
making distinct bids on different categories 
of search query. The challenge for the advertiser
is to balance her bids across the range of categories
offered to her by the platform.

Suppose the platform allows advertiser $i$ to 
partition the type space $\mT$ into categories $(\mT_{ik}: k\in\mK_i)$.
The categories may be defined in terms of the keywords
used in a search or any other feature of the search type,
such as geographical area or broad classification of the user, 
that the platform is prepared to share with advertiser $i$.
We assume the platform allows  advertiser $i$ to know
the category of the search type $\tau$, namely 
that $\tau \in\mT_{ik}$, but the
platform knows more, namely $\tau$. 
We suppose the platform may vary aspects
of the auction, such as the number of
advertising slots on the page or more generally the
layout of the page, depending on the search type $\tau$.
For example, the platform may use the current screen size of
the user to determine the page layout.

Let $b_{ik}$ be the bid of advertiser $i$ for click-throughs from 
category $k$, and let $b_i
 = (b_{ik} :  k \in \mathcal{K}_i)$ and $b
 = (b_{ik} : i\in\mI,  k \in \mathcal{K}_i)$.
Let $y_{ik}$ be the click-through
rate to advertiser $i$ from searches in category $k\in\mK_i$,
and assume that the expected rate of payment by advertiser $i$
for click-throughs from category $k\in\mK_i$ is 
\begin{equation*}  
\pi_{ik} (b) y_{ik}(b)
= \int_0^{b_{ik}}
\Big({y}_{ik}(b)-{y}_{ik}(b'_{ik}, b) \Big)
db'_{ik},  \qquad i\in\mI, \; k\in\mK_i, 
\end{equation*}
where
$(b'_{ik}, b)$ is the vector obtained from the 
vector $b$ by replacing the component $b_{ik}$
by $b'_{ik}$. 
This rate of payment can be achieved by  either 
of the first two pricing implementations
of Section~\ref{Sec:pricing}:  these implementations
use the function $y_i^\tau(b)$ to determine the charge for
a click-through, and so no difficulty is
caused by the form of the auction depending upon the 
search type $\tau$.

Let $y_i= (y_{ik}: k \in \mathcal{K}_i)$. 
If advertiser $i$'s utility $U_i(y_i)$ is simply a  sum of 
utilities $U_{ik}(y_{ik})$ over the categories $k \in \mathcal{K}_i$
then this model is subsumed in the model treated in earlier sections:
advertiser $i$ can be represented by a collection of 
sub-advertisers, one for each category $k \in \mathcal{K}_i$,
and the platform can set click-through
probabilities to zero for sub-advertiser $k \in \mathcal{K}_i$
if $\tau \notin \mT_{ik}$. But for more general utility functions
we would expect that the bids $b_{ik}, k \in \mathcal{K}_i$, 
cannot be determined independently. 

Suppose, then,  that advertiser $i$'s utility 
$U_i(\cdot)$ is  an
increasing, strictly concave, continuously differentiable
function of the vector $y_i= (y_{ik}: k \in \mathcal{K}_i)$.
Assume that 
the partial derivative ${\partial U_i}/{\partial y_{ik}}$ decreases from $\infty$ to $0$ as $y_{ik}$ increases from $0$ to $\infty$, and that 
$b_{ik}\mapsto {y}_{ik}(b_{ik}; b )$
satisfies the monotonicity property.

Let
\begin{equation*}
U^*_i(b_i)=\max_{y_i\geq 0} \left( U_i(y_i)
-\sum_{k \in \mathcal{K}_i} b_{ik}   {y}_{ik} \right),
\end{equation*}
the
Legendre-Fenchel transform of $U_i(y_i)$, interpretable
as the consumer surplus of advertiser $i$ at prices $b_i$.
Our conditions on $U_i$ and its partial derivatives ensure
there is a unique maximum, interior to the positive orthant,
for any price vector $b_i$ in the positive orthant.
Let $(D_{ik}(b_i): k \in \mathcal{K}_i)$
be the argument $y_i$ that attains
this maximum: it is the demand vector of advertiser $i$
at prices $b_i$, and
\begin{equation} \label{partialmult}
\frac{\partial } {\partial b_{ik}} U^*_i(b_i)
= - D_{ik}(b_i) . 
\end{equation}

Then the question
for advertiser $i$ is how to balance her bids 
$(b_{ik}: k \in \mathcal{K}_i)$
over the
categories  $\mathcal{K}_i$ that are of interest to her. 
The payoff to advertiser $i$ arising from
a vector of bids
$b= (b_{i} : i\in\mI) = (b_{ik} : i\in\mI,  k \in \mathcal{K}_i)$ is then
\begin{equation*} 
u_i(b) = U_i ({y}_i(b))  - \sum_{k \in \mathcal{K}_i} 
\pi_{ik}(b) {y}_{ik}(b), 
\end{equation*}
and the condition for a Nash equilibrium is again~\eqref{Nashr} where
now $b_i$ is a vector.  
Paralleling the development of 
Section~\ref{Mech}, the maximum of the payoff 
function $b_i\mapsto u_i(b_i,b_{-i})$ is attained when
\begin{equation*} 
\frac{\partial }{\partial  b_{ik} } U^*_i(b_{i})
 + y_{ik}(b) = 0, \quad  k \in \mathcal{K}_i,
\end{equation*}
or equivalently $D_{ik}(b_i) = y_{ik}(b)$ for $k \in \mathcal{K}_i$, 
there is a unique Nash equilibrium, 
and these conditions also 
identify the unique system
optimum.

Next suppose that for each $  k \in \mathcal{K}_i$  advertiser $i$ changes her 
bid $b_{ik}(t)$ smoothly
(i.e., continuously and differentiably)
as a consequence of her observation of her 
current click-through rate $y_{ik}(b(t))$
so that
\begin{equation} \label{accord}
\frac{d}{dt}  b_{ik}(t) \gtrless 0 \text{ according as }
y_{ik}(b(t)) \lessgtr D_{ik}(b_i(t)).
\end{equation}
This is a dynamical system representation of advertiser $i$ varying 
$b_{ik}$ smoothly in order to increase or decrease her bid
for keyword $k$ according  to whether the 
currently observed click-through rate $y_{ik}(t)$ is lower or higher
than her demand at her current bid prices. 
Then trajectories converge to the solution of the system problem,
by essentially the same Lyapunov argument as used to prove
Theorem~\ref{CONV}, as we now sketch.

Let 
\begin{equation*}
\mathcal V(b )  = \sum_{i\in\mI} U_i^*(b_i) 
+ \sum_{i\in\mI} \sum_{k \in \mathcal{K}_i}  b_{ik}
{y}_{ik}(b) . 
\end{equation*}
Differentiating $\mV(b(t))$ yields, from~\eqref{partialmult},
Lemma~\ref{diffsumyprop}
and~\eqref{accord}, 
\begin{align*}
\frac{d}{dt}\mathcal V(b(t))
&= \sum_{i\in\mI} \sum_{k \in \mathcal{K}_i} \frac{\partial \mathcal
V}{\partial b_{ik}}
\frac{d}{dt} b_{ik}(t) = -\sum_{i\in\mI} \sum_{k \in \mathcal{K}_i}
\left( D_{ik}(b_i(t)) -
y_{ik}(b(t)) \right)
\frac{d}{dt} b_{ik}(t) \leq 0
\end{align*}
where the inequality is strict unless
$D_{ik}(b_i(t)) = {y}_{ik}(b(t))$ for $i\in\mI, k \in
\mathcal{K}_i$.
But this holds if and only if $y$ solves the system problem.

\begin{remark}
Our approach to auction design separates the computational burden into a
task that can be completed quickly for each page impression by the
platform, and tasks that can be performed more slowly by individual
advertisers or perhaps agents working on their behalf. The task 
for an advertiser is to assess the value to her of 
different forms of click-through. This may not be an easy task,
but it is a task naturally assigned to the advertiser and is made simpler
by requiring only local information in the region of the currently achieved
click-through rates. 
\end{remark}

\begin{remark}
How finely should a search platform allow categories
to be defined, and how finely should it divide its
stream of queries across distinct auctions?
Finer classifications will allow
advertisers to communicate more precisely their
valuations but excessive targeting may lead to
thinner markets and to various forms of adverse selection. These trade-offs
are discussed by \cite{levinmilgrom}, who argue that
the degree of differentiation allowed, or conflation imposed, is
an important aspect of the organization of well-functioning
markets.

In the current context,
observe that advertiser $i$ is forced to conflate her bid $b_{ik}$
across multiple auctions, and in each of these
the set of competing advertisers is likely to be different. Thus
the design of the categories $(\mT_{ik}: k\in\mK_i, i \in \mI )$
provides
ample opportunity to balance the degree of
differentiation allowed to, or conflation imposed upon,  advertisers
by the platform.
\end{remark}

We end this section with two examples which indicate the connections
between our work and earlier
important approaches in the traffic engineering and
resource allocation literature. 

\begin{example} \label{weightings}
Consider a platform with advertisers who prefer
click-throughs that come from one geographical 
area rather than another, or from one set of keywords
rather than another, 
simply because such click-throughs are more likely to 
convert into sales. Then advertiser $i$'s utility
will be a univariate function 
\begin{equation}\label{Uweights}
U_i\Big(  \sum_{k\in\mK_i} w_{ik} y_{ik}\Big)
\end{equation}
where we assume 
$U_i(\cdot)$ satisfies our earlier assumptions
from Section~\ref{sec:opt} and where $w_i = ({w}_{ik}: k\in\mK_i)$ account
for the weight applied to each category by advertiser $i$.

Given that advertiser $i$ 
declares a bid $\tilde{b}_i$ and weights $\tilde{w_i}=(\tilde{w}_{ik}: 
k\in\mK_i)$ (strategically and not necessarily equal to $w_i$), the platform
may use
the information contained in $\tilde{w}$ as well as $\tilde{b}$ and $ \tau$ 
to solve the revised
assignment problem, \text{ASSIGNMENT}($\tau,\tilde{b}, \tilde{w}$),
defined as problem~\eqref{ASMT} with the revised
objective:
\[
{\rm Maximize}  \quad \quad \sum_{i\in\mI} \tilde{b}_i  \sum_{l\in\mL}
\tilde{w}_i^\tau
p_{il}^\tau
x_{il}^\tau,
\]
where $\tilde{w}^\tau_{i} =\tilde{w}_{ik}$ for $\tau\in\mT_{ik}$ and
$k\in\mK_i$. 
Write $b_{ik}=\tilde{b}_i \tilde{w}_{ik}$ and
 $b = (b_{ik} : k \in \mathcal{K}_i, i\in\mI)$.
then the model formally reduces to that of this section.
Essentially there
are number of parallel auctions taking place with 
the search type $\tau$ determining the bids of advertisers: 
if $\tau \in \mT_{ik}$ then the bid of 
advertiser $i$ is $b_{ik}$. 

The reduction replaces
the vector function $U_i(y_{ik}: k \in \mK_i)$ 
by the special case~\eqref{Uweights},
a univariate function of a weighted sum. 
The utility~\eqref{Uweights} is concave but not strictly concave,
and this introduces some minor technicalities and some simplifications.
When strategic advertisers 
maximize over $(b_{ik}: k\in\mK_i)$ their respective payoff 
functions
\[
{u}_i(b) :=  U_i\Big(  \sum_{k\in\mK_i} w_{ik} y_{ik}(b) \Big) -
\sum_{k\in\mK_i}\pi_{ik}(b) y_{ik}(b)
\]
the resulting Nash equilibrium is achieved  under the necessary conditions
\begin{subequations} \label{Wardrop} 
\begin{align}
w_{ik} U_i'\Big(  \sum_{k \in\mK_i} w_{ik} y_{ik}(b)\Big) = b_{ik}, &\qquad
\text{if}\quad  y_{ik}(b) > 0,\label{weightNash1}\\
w_{ik} U_i'\Big(  \sum_{k\in\mK_i} w_{ik} y_{ik}(b)\Big) \leq b_{ik},
&\qquad\text{if}\quad  y_{ik}(b) = 0\label{weightNash2},
\end{align}
\end{subequations}
for $k\in\mK_i$ and $i\in\mI$.
In particular, \eqref{weightNash1} implies that if $y_{ik}(b) > 0$ then
$w_{ik}/b_{ik}$ does not depend on $k\in\mK_i$: hence 
the weights intrinsic to advertiser $i$, $(w_{ik}: k\in\mK_i)$, and the
weights declared strategically by advertiser $i$ 
to the platform, $(\tilde{w}_{ik}=b_{ik}/b_i: k\in\mK_i)$, are in
proportion wherever a positive click-through rate is received (this is 
the simplification achieved by the form~\eqref{Uweights}). 
So the auction mechanism achieves incentive compatibility: 
advertisers are encouraged to truthfully declare their intrinsic weights
for categories for which they are competing.

We can interpret the system optimization as an infinitely large
bipartite congestion game, an interpretation that parallels one of
Vickrey's early motivations in transport pricing;
in particular the  conditions~\eqref{Wardrop} parallel the conditions
for  a traffic equilibrium in a network, \cite{wardrop1952road}, 
\cite{Beckmann1956}.

\end{example}

\begin{example}  \label{budget} 
An advertiser may have a budget constraint
on what she can spend across different types of search, for example,
in an advertising campaign. In this example we note
a simple approach
which captures a
budget constraint within the framework of this section. 

Suppose
\[
U_i (y_i) =  \frac{B_i}{q} \log \sum_{k \in \mathcal{K}_i} (w_{ik}y_{ik})^q
\]
for $0 < q  <1$. Then
\[
 \frac{\partial U_i}{\partial y_{ik}} = \frac{B_i w_{ik}^q{y_{ik}}^{q-1}}
{\sum_{j \in \mathcal{K}_i} (w_{ij}y_{ij})^q}
\]
and so at the unique Nash equilibrium described earlier
in this section, where $\partial U_i / \partial y_{ik} =
b_{ik}$, the budget constraint
\[
\sum_{k \in \mathcal{K}_i} b_{ik} y_{ik}  = B_i
\]
is automatically satisfied; note that the constraint is
on the rate of bidding rather
than expenditure, i.e., not taking into account rebates.

We require $q  <1$ to ensure the strict concavity of $U_i(\cdot)$. As
$q \to 1$, maximizing $U_i(y_i)$ subject
to the budget constraint approaches the problem of maximizing
$\sum_{k \in \mathcal{K}_i} w_{ik}y_{ik}$ subject
to the same budget constraint. In the special case where
all advertisers are
budget constrained  we recover an important early
model for the equilibrium price of goods for buyers
with linear utilities, \cite{Fi1892,EiGa59}.

Advertiser $i$ will receive a stream of rebates, which
may be delayed and will be noisy.
Rebates will cause
the total spent in a period to be less than the budget $B_i$,
and a natural control response would be to spread the total rebate received
in one time period over the budgets available for later time periods.
\cite{Borgs:2007:DBO:1242572.1242644} explore a natural
bidding heuristic for a budget-constrained advertiser which
readily extends to include delayed rebates.

\end{example}

%% file: RELATED_FINAL.tex

\section{Related Work} \label{sec:related}
In this paper we have considered a problem where the social welfare of an
auction system is optimized subject to the capacity constraints of that
system. 
Social welfare optimization has long been an objective in the design of effective market mechanisms, \cite{Vick61}. 
However, only in the recent literature have computationally efficient methods
been considered for market and auction design, see  \cite{BDX11},
\cite{JaVa07}, \cite{Vaz06}. In the context of electronic commerce and
specifically sponsored search auctions, these computational considerations
are of critical importance given the increased diversity and competition
associated with online advertising.

We have  applied a decomposition approach to the task of optimizing advert
allocation over the vast range of searches that can be conducted,
and separated the task into  sub-problems which can be implemented
by each advertiser and on each search.
The decompositions of interest are familiar and have been important 
in the context of communication network design, \cite{Sr04, KY14}. 
\cite{srta10} is a distinct approach using optimization decomposition
ideas, but based instead on a queueing model of an on-line advert campaign
and using connections to scheduling in wireless networks.

Strategic formulations of these optimization decompositions have  been 
developed: in a simple model 
\cite{johari2004efficiency} show a price of anarchy of 75\% at
a Nash equilibrium. Notably, a single parameter VCG mechanism to yield efficient
allocation was considered by \cite{maheswaran2004social} and 
subsequently generalized by \cite{YaHa07} and \cite{jots09}. Here a
parametrized surrogate utility is employed in the VCG mechanism, where the parameter is selected strategically. 
The message passed from the player to the mechanism is thus 
chosen from a reduced space. 
One part of our decomposition can be viewed  as deriving a  single
parameter VCG mechanism with linear utilities; that is, despite allowing
more general utility functions, the derived mechanism is \textit{as if}
each player had a linear utility function. 
Prior works have used strictly concave surrogate functions while in our
approach linearisation is possible owing to the large search/constraint space employed.
A linear VCG allocation is computationally straightforward 
(\cite{leonard1983elicitation}, \cite{BdVSV}), 
but the 
crucial advantage of our linear framework is that -- in addition to decomposing
the objectives of  advertisers and the platform -- further decomposition
over the search space is possible, leading to a practical 
mechanism.
In particular, the mechanism can be implemented on each search instance:
allocation and pricing both involve standard polynomial time algorithms per-search and per-click, respectively.
In essence, we find a simple implementable auction mechanism that yields an efficient allocation of adverts across the entire search space.

Parametrized VCG mechanisms are examples
of simplified mechanisms, where the set of messages available
to report preferences is restricted.  
 \cite{milgrom2010simplified} has shown that the equilibria of 
a simplified mechanism  are also equilibria of the 
unrestricted mechanism when a
certain outcome closure property is satisfied. 
The closure property states that a
bidder can make an optimal best response within the set of restricted bids
whenever other bidders' choices lie within the restricted set. 
As an example, the closure
property can be applied to concave utility functions under the
restriction of linearity, with the restricted bid communicating a 
tangent plane rather than the entire utility function.

A very common simplification applied in sponsored search auctions 
is conflation.
For example advertisers may be required to make the
same bid per click whatever the position of an advert on 
the page. If advertisers differentiate
between positions beyond each position's observed click-through rate
(for example, if click-throughs from lower positions are
less or more valuable to the advertiser), then there may be
a loss of social welfare from the 
restriction that a bidder must communicate a single parameter to a
mechanism which is unaware
of these positional effects. 
The question of whether 
VCG or GSP mechanisms with this restriction are sufficiently
expressive to communicate the bidders' true values
for positions 
is discussed in detail in 
\cite{milgrom2010simplified} and \cite{dutting2011simplicity}. 
\cite{aggarwal2007bidding} discuss mechanisms which maintain an efficient equilibrium by allowing bidders to specify a minimum slot, in addition to their bid.
Further recent discussions of simplified mechanisms for
sponsored search auctions are 
\cite{chawla2013auctions,hoy2013dynamic,
dutting2014expressiveness}. 

The context in these papers is, in our terms, an auction 
for a single type of search. 
This context serves to 
illustrate an important theoretical question concerning
simplified mechanisms. 
The diverse
stochastic variability found 
in the sponsored search market (see \cite{pk:11}) 
makes the assumption of a single type of search unrealistic. 
The framework we adopt is
rather different: we presume, as described in Section~\ref{sec:assign}, 
that the search platform
knows more than the advertiser about the type of search
being conducted, for example about the searcher,
and that this information affects click-through probabilities. 
For the advertiser there is therefore a considerable
further conflation: the same bid for a click-through
is used over an entire category of search query as well as over
different positions. The information asymmetry between
the platform and the advertisers allows the platform
to assign and price adverts using a per-search level
of granularity on the search type, while the advertisers 
experience only average click-through rates over a diverse
set of search types. Our framework is designed to model
advertisers who differentiate the value of a click-through 
according to search categories, defined in terms of keywords and user
characteristics, rather than advert position: see Section~\ref{sec:ext}.

\cite{Borgs:2007:DBO:1242572.1242644} gave an 
important early treatment of the dynamics of bid optimization, 
and emphasised the importance of equalizing
the ``return-on-investment" across keywords
for budget-limited advertisers. These authors also used a 
continuous time formulation, noted the importance
of random perturbations, proved convergence to a market equilibrium 
in the case of first price auctions and observed
experimentally convergence in the case of second price
auctions. Auctions were all single slot, and bids were
assumed truthful. 
The additional contribution of Section~\ref{Dynamics} 
to this early work on dynamics is that we have established
convergence 
for a  wide class of continuous time dynamics
representing advertisers' best response under our pricing mechanism. 

More recent work on learning and bid optimization is 
reviewed by \cite{eps365566}, who
use the framework
of a multi-armed bandit to devise policies that maximize the expected
number
of click-throughs in a given number of searches within a given budget.
This stream of research typically uses no-regret learning, expressing
convergence in terms of sub-linear temporal convergence. 
\cite{iyer2011mean}
use a mean field approach to treat agents who
need the learn the value to them of a click-through.
These are challenging problems, even for the sequence of single slot second
price auctions treated in these papers.
By comparison our approach uses a dynamical systems framework where fluid
averages are controlled.  We are  also able to  show convergence to a Nash equilibrium, rather than a correlated equilibrium.   Our approach deliberately simplifies the
modelling of the stochastic
streams of click-throughs, which it represents with just their means, but
is
able to deal with multi-slot auctions and with streams of click-throughs
arising from different keywords and categories of searcher.
\cite{Nekipelov2015EC} review work on learning, focusing on a
model of sponsored search auctions: their analysis of data from
BingAds indicates that
typical advertisers  bid a
significantly shaded version of their value, as would be
expected in a  generalized second price auction rather
than a VCG auction.

As noted by \cite{milgrom2010simplified}, the most devastating objections
to Vickrey pricing (\cite{ausubelmilgrom}) apply only when bidders can buy multiple items, and have
no force in sponsored search auctions where each bidder can acquire at most
one position.
\cite{varian2014vcg} have recently argued 
that VCG mechanisms are of practical
interest because they are flexible and extensible. For this reason,
Facebook implements
VCG\footnote{https://developers.facebook.com/docs/marketing-api/pacing -
downloaded on 20 August 2015.} rather than the generalized second price auction
currently used by Bing and Google. These considerations are particularly
relevant for contextual advertisement, unordered page layouts, image-text
adverts and image-video adverts. Such extensions are important and are of
growing interest to online advertisement platforms; see \cite{Goel} for a
recent discussion.

%% file: CONCLUDE_FINAL.tex

\section{Concluding Remarks} \label{sec:conclude} 

We describe a framework to capture the system
architecture of Ad-auctions. The assignment problem
must be solved rapidly, for each search, while
an advertiser is primarily interested in aggregates
over longer periods of time. The platform knows more about 
the type of a search query
and thus more about click-through probabilities, while
an advertiser knows more about the value to her of additional
click-throughs and is incentivized to communicate
this information via her bids. Thus we model in detail each random
instance of the assignment problem, while we describe
an advertiser's strategic behaviour in terms of averages evolving
in time.  On a slow time-scale
the platform may decide which search types to
pool in distinct auctions, across which the
advertisers will have different preferences
they are able to communicate.

We have used sponsored search auctions as the motivation,   and our model
reflects current practice in sponsored search, where platforms such as
BingAds or Google Adwords  use a variant of the second price
auction to solve the assignment and pricing problem for every search query,
while  advertisers alter bids on timescales measured in hours or days. The
setting, allowing for a large, continuous range of search types and
varying competition, greatly extends the scope of prior models which are
typically limited to the auction of a single keyword amongst a static pool
of advertisers.

We address the task of achieving efficiency over all a platform's searches
under a pay-per-click pricing model. Under the assumption of strategic advertisers, we showed that, with
appropriate pricing, a Nash
equilibrium exists for the advertisers which achieves welfare 
maximizing assignments. 
We  gave  an  appealingly simple way to
implement these prices: namely, by giving advertisers a rebate, constructed
by solving a second assignment problem.
The first assignment is implemented with low computation cost and the solution to the second assignment problem is not used for the
allocation but only for pricing. Further this mechanism is found to be flexible and extends in a straightforward manner to various different page-layouts. Hence
 under the pay-per-click model, this mechanism shows potential to be adapted for use in
current Ad-auctions.

%% file: APPENDIX_FINAL.tex

\section{Proof of Propositions \ref{diffcont} and Lemma \ref{diffsumyprop} }\label{Lemma1}
This section gives proofs of Proposition \ref{diffcont} and 
Lemma \ref{diffsumyprop} concerning properties
of 
the functions
\begin{align*}
{y}_{i}(b)=\bE_{\tau} \sum_{l} p^\tau_{il} x_{il}^{\tau}(b),\qquad \sum_{i} b_i {y}_i(b).
\end{align*}
Proposition \ref{diffcont} requires the further technical lemma, Lemma \ref{Lemma1.1}, which give the Lipschitz continuity of a random point belonging to a polytope as we smoothly change the description of its facets. Lemma \ref{diffsumyprop} employs the Envelope Theorem \cite[Chap. 3]{milgrom2004putting}, as is commonly applied in auction theory.

\begin{lemma}\label{Lemma1.1}  1) If  $U$ is a random vector uniformly distributed inside the unit sphere, $S_n=\{ u\in\bR^n  : ||u||\leq 1 \}$, 
then there exist a constant $K_1$ such that for any two non-zero vectors $b,\tilde{b}\in \bR^n\backslash \{ 0\}$
\begin{equation*}
\bP( b^\T U \geq 0 > \tilde{b}^\T U ) \leq \frac{K_1}{||b || \wedge ||\tilde{b} ||} || b - \tilde{b}||.
\end{equation*}

\noindent 2) 
If $X$ is a random variable with density $f_X$ continuous on its support $\mP$, a polytope $\mP\subset [-1,1]^n$,
then the function $\bP( \mu_1^\T X \geq 0, ...,\mu_k^\T X \geq 0 )$ is Lipschitz continuous as a function of $\mu_1,...,\mu_k$ provided $||\mu_1||,...,||\mu_k||$ are bounded away from zero. 
\end{lemma}
 \begin{proof}
1) We give a geometric proof of the result. We assume, wlog, that $||b||\geq ||\tilde{b}||$, and we let $V_n$ be the volume of $S$. 
For every $u$ satisfying $b^\T u \geq 0 > \tilde{b}^\T u$, there exists a
$\theta\in [0,1]$ such that $b^\T u + \theta ( \tilde{b}^\T - b^\T)u=0$.
Let $b_\theta$ be the unit vector proportional to $b + \theta ( \tilde{b} -
b)$. So, $b_{\tilde{\theta}}^\T u =0$. Or, in other words, if $b^\T u +
\theta ( \tilde{b}^\T - b^\T)u=0$ then $u$ belongs to a cross section of
the sphere $\{ u' \in \mS : b_\theta u' =0\}$ for some $b_\theta$
proportional to $b + \theta( \tilde{b} -b)$. Thus the volume of $\{ u :b^\T
u \geq 0 > \tilde{b}^\T u \}$ is overestimated by the area of the sets $\{ u' \in \mS : b_\theta u' =0\}$ multiplied by the change in the normal vector $b_\theta$.

With this in mind we note three facts: 1) Each cross section $\{ u\in S :  b_\theta^\T u =0\}$ has the same volume $V_{n-1}$ in its ${n-1}$ dimensional subspace; 2) The path $\mP=\{ b_\theta : \theta\in [0,1]\}$ is a circular path starting at $b/ || b||$ and ending at $\tilde{b}/ || \tilde{b}||$, and thus has length bounded above by the terms 
\begin{equation*}
2\pi \bigg|\bigg| \frac{b}{||{b}||} - \frac{\tilde{b}}{||\tilde{b}||}\bigg|\bigg|  \leq \frac{2\pi}{||{\tilde{b}}||}  \big|\big|   b - \tilde{b} \big|\big|; 
\end{equation*}
and, 3) $\{u\in S: b^\T u \geq 0 > \tilde{b}^\T u\} = \{ u\in S :  b_\theta^\T u =0,\; \theta\in [0,1]\}$. Thus, we see we can bound the probability $\bP( b^\T U \geq 0 > \tilde{b}^\T U )$ by the length of the path $\mP$ times the volume of cross sections $\{ u\in S :  b_\theta^\T u =0\}$. In other words,
\begin{equation*}
\bP( b^\T U \geq 0 > \tilde{b}^\T U ) \leq \frac{2\pi V_{n-1}}{||\tilde{b}||} || b - \tilde{b}||,
\end{equation*}
as required.\\
2)  A function which is componentwise Lipschitz continuous  is Lipschitz continuous. So, without loss of generality, we prove that the first component of our function is Lipschitz continuous. Observe
\begin{align}
& \Big|\Big| \bP \left( \mu_1^\T X \geq 0,  \mu_2^\T X \geq 0, ...,\mu_k^\T X \geq 0\right) - \bP\left( \tilde{\mu}_1^\T X \geq 0 ,  \mu_2^\T X \geq 0, ...,\mu_k^\T X \geq 0\right)  \Big|\Big| \notag \\
= & \Big|\Big| \bP \left( \mu_1^\T X \geq 0 > \tilde{\mu}^\T X,  \mu_2^\T X \geq 0, ...,\mu_k^\T X \geq 0\right) - \bP \left( \tilde{\mu}_1^\T X \geq 0 > \mu^\T X ,  \mu_2^\T X \geq 0, ...,\mu_k^\T X \geq 0 \right) \Big|\Big| \notag\\
\leq &  \bP \left( \mu^\T X \geq 0 > \tilde{\mu}^\T X\right) + \bP \left( \tilde{\mu}^\T X \geq 0 > \mu^\T X \right)\, . \label{PLipBound}
\end{align}
Also since $f$ is a continuous  density on click-through probabilities $\tilde{\mP}$, it is bounded by a constant. So, we can bound the above probabilities with uniform random variables:
\begin{equation*}
 \bP \left( \mu^\T X \geq 0 > \tilde{\mu}^\T X\right) \leq K_2  \bP \left( \mu^\T U \geq 0 > \tilde{\mu}^\T U\right)
\end{equation*}
for a constant $K_2$ and for $U$ a uniform random variable on the unit sphere in $\bR^n$. Now applying part 1) of this Lemma
\begin{align*}
 \bP \left( \mu^\T X \geq 0 > \tilde{\mu}^\T X\right) \leq &\frac{K_1 K_2}{ ||\mu||\wedge || \tilde{\mu}||}   || \mu - \tilde{\mu}|| 
\leq \frac{K_1 K_2}{ K_3}   || \mu - \tilde{\mu}|| .
\end{align*}
where $K_3$ is the constant by which $\mu$ and $\tilde{\mu}$ are bounded away from zero.
Thus, applying this inequality to \eqref{PLipBound}, we have that $ \bP \left( \mu_1^\T X \geq 0,  \mu_2^\T X \geq 0, ...,\mu_k^\T X \geq 0\right)$ is Lipschitz continuous in its first component and thus is Lipschitz continuous.
   \end{proof}

The previous lemmas suggest that provided there is a certain amount of variability in $p_{ij}^\tau$ then we can expect the average performance of an advertiser to be a continuous function of the declared prices $b$. 

\begin{proof}[Proof of Proposition \ref{diffcont}]
First we argue the continuity of $b_i\mapsto y_i(b_i,b_{-i})$ and then argue that it is strictly increasing and positive. We let $\mS$ index the assignments that can be scheduled from $\mI$ to $\mL$.
Notice, provide there is a unique maximal assignment,
\begin{align}\label{xeq}
x^\tau_{il}(b)
&=\sum_{\pi\in\mS: \pi(i)=l} \bI \left[ \sum_{k}b_kp^\tau_{k\pi(k)} \geq \sum_k b_k p^\tau_k \tilde{\pi}(k),\; \forall \tilde{\pi}\neq \pi \right]\notag \\
& =\sum_{\pi\in\mS: \pi(i)=l} \prod_{\tilde{\pi} \neq \pi } \bI \left[ \sum_{k}b_kp^\tau_{k\pi(k)} \geq \sum_k b_k p^\tau_{k \tilde{\pi}(k)}\right].
\end{align}
Here $\bI$ is the indicator function.  Notice, since $\bP_{\tau}$ admits a density, $f(p^\tau)$, then with probability one there is a unique maximizer to the  problem ASSIGNMENT($\tau,b$). 
So the equality \eqref{xeq} holds almost surely for all $b> 0$.

For two assignments $\pi$ and $\tilde{\pi}$, we define the vector
\begin{equation*}
\mu_{\pi \tilde{\pi}} := \left( b_i \bI[\pi(i)=l] - b_i \bI[\tilde{\pi}(i)=l]   :\; i\in \mI, l\in\mL \right).
\end{equation*}
Notice for any two distinct permutations, the non-zero components of $\bI[\tilde{\pi}(i)=l]$ are distinct. So  the vectors $\mu_{\pi \tilde{\pi}}$ are distinct and non-zero over $\tilde{\pi}\neq \pi$.
Since the maximal assignment is almost surely unique, we have
\begin{align}
{x}_{il}(b)&=\sum_{\pi\in\mS: \pi(i)=l}  \bE\left[  \prod_{\tilde{\pi} \neq \pi } \bI \left[ \sum_{k}b_kp^\tau_{k\pi(k)} \geq \sum_k b_k p^\tau_{k \tilde{\pi}(k)}\right]\right] = \sum_{\pi\in\mS: \pi(i)=l}  \bP\left( \mu_{\pi \tilde{\pi}}^\T p \geq 0,\; \forall \tilde{\pi} \neq {\pi}\right). \label{xbarY}
\end{align}
Thus if the function $  \bP\left( \mu_{\pi \tilde{\pi}}^\T p \geq 0,\; \forall \tilde{\pi} \neq {\pi}\right) $ is  Lipschitz continuous then we have same properties for functions ${x}_{jl}(b)$. The Lipschitz continuity of $ \bP\left( \mu_{\pi \tilde{\pi}}^\T p \geq 0,\; \forall \tilde{\pi} \neq {\pi}\right)$ is proven in Lemma \ref{Lemma1.1}.
This implies the Lipschitz property for ${x}_{il}(b)$ with $b>0$ and since $y_i$ is a finite sum of these terms the same continuity holds for $b_i \mapsto y_i(b_i,b_{-i})$, with $b=(b_i,b_{-i})>0$.  
Further, continuity at $b_i=0$ is also ensured by bounded convergence: there are greater than $|\mL|$ positive bids occur in $b_{-i}$ the assignment of these must eventually outweighs the assignment of $i$ as $b_i\searrow 0$. In other words, \eqref{xeq} goes to zero point-wise as $b_i\searrow 0$. Thus bounded convergence applies to \eqref{xbarY} and, also, $y_i(b)$ which implies $y_i(b_i,b_{-i})\rightarrow 0$ as $b_{i}\searrow 0$.

We now prove that the function $b_i\mapsto {y}_i(b_i,b_{-i})$ is strictly increasing for $b_{-i}\neq 0$. First we show that it is increasing. 
Since ${y}_i(b) = \bE_\tau y^\tau_i(b)$  \eqref{ASMT9}, if we can prove $y^\tau_i(b)$ is increasing then so is ${y}_{i}(b)$. Further note
\begin{equation*}
\sum_{i\in\mI} b_i y_i^\tau(b) = \sum_{i\in\mI, l\in\mL} b_i p_{il}^\tau x_{il}^\tau(b)
\end{equation*}
which is the optimal objective for the assignment problem \eqref{ASMT}.

Define $b'$ with $b_i' < b_i$ and $b_j' = b_j$ for each $j\neq i$. We now proceed by contradiction. Suppose that ${y}_i(b') > {y}(b)$, then the following equalities and inequalities hold
\begin{align*}
\sum_{j\in\mI} b_j y^\tau_j(b) &= (b_i - b_i' ) y^\tau_i(b) + \sum_{j\in\mJ} b'_j y^\tau_j(b) \\
& \leq (b_i - b_i' ) y^\tau_i(b) + \sum_{j\in\mJ} b'_j y^\tau_j(b') \\
& <  (b_i - b_i' ) y^\tau_i(b') + \sum_{j\in\mJ} b'_j y^\tau_j(b') =\sum_{j\in\mJ} b_j y_j^\tau(b').
\end{align*}
Here the first equality holds by the optimality of $y^\tau(b')$ and the second holds by assumption. But notice the resulting equality above contradicts the optimality of $y^\tau(b)$. Thus by contradiction, $y_i^\tau(b)$ is increasing in $b_i$ and, after taking expectations, so is ${y}_i(b)$. 

We now prove that $b_i \mapsto {y}_i(b)$ is strictly increasing. Let $b'$ be such that $b'_i > b_i$ and $b'_j = b_i$ for all $j\neq i$. The result proceeds by showing that
\begin{equation*} 
\bP( y^\tau_i(b') >  y^\tau_i(b) | E ) > 0
\end{equation*}
where we condition on an event $E$ with non-zero probability. Notice, after taking expectations, this implies that ${y}_i(b') > {y}_i(b)$.

Now since $f(p)$ is positive on a region containing the origin, $f(p)$ stochastically dominates a uniform random variable on the set of increasing click-through rates, $\tilde{\mathcal{P}}\cap [0,\epsilon]^{\mI\times \mL}$, for some $\epsilon$. Thus it is sufficient to prove the result for $u=(u_{il} : i\in\mI, l\in\mL )$ uniform on $\tilde{\mathcal{P}}\cap [0,\epsilon]^{\mI\times \mL}$. Now, for instance, there is positive probability 
that advertiser $i$ and $j$, with $b_j>0$, compete exclusively over the top
two slots, $l=1,2$. This occurs, for instance, when $i$ and $j$ have
click-through rate over $ \epsilon /2 $ and all other advertisers have expected revenue that is half of the lower bound revenue of $i$ and $j$, namely, the event 
\begin{equation*}
E := \left\{ \min_{\substack{k= i,j\\l\in\mL}}  u_{kl} \geq \frac{\epsilon}{2},\;\; 2 \max_{\substack{k\neq i,j\\l\in\mL}}\left\{ b_k u_{kl} \right\} \leq \frac{\epsilon}{2} \min\{ b_i , b_j  \} \right\}.
\end{equation*}
This event has positive probability and then only $i$ and $j$ can appear on the top two slots on this event.

Given this event, advertiser $i$ achieves the top position with bid $b'_i$ and the second position with bid $b_i$ on the condition
\begin{equation*}
b_i' (u_{i1}-u_{i2} ) > b_j (u_{j1}-u_{j2}) > b_i (u_{i1}-u_{i2} ).
\end{equation*}
Since, after conditioning on $E$, $u_{i1},u_{i2},u_{j1},u_{j2}$ remain
independent and uniformly distributed (on the set $\tilde{\mathcal{P}}\cap \{ u_{i1},u_{i2},u_{j1},u_{j2} \geq \epsilon/2 \}$), it is a straightforward calculation that
\begin{equation*}
\bP \left( b_i' (u_{i1}-u_{i2} ) > b_j (u_{j1}-u_{j2}) > b_i (u_{i1}-u_{i2} ) \big|  E \right) >0.
\end{equation*}
Since $u_{i1}$, the value of $y_i^\tau$ achieved by $b'_i$ on $E$, is
strictly bigger than $u_{i2}$, the value of $y_i^\tau$ achieved by $b_i$ on $E$, the above inequality implies
\begin{equation*}
\bP( y^\tau_i(b') >  y^\tau_i(b) | E ) > 0,
\end{equation*}
and thus ${y}_i(b)< {y}_i(b')$, as required. Further, note that this argument implies the required property that $y_i(b_i,b_{-i})>0$ for $b_i>0$.
   \end{proof}

\begin{lemma}\label{diffsumyprop}
The function $b \mapsto \sum_{i\in\mI} b_i {y}_i(b)$ is convex
and continuously differentiable for $b\neq 0$; further, 
\begin{equation*} 
\frac{d}{d b_i}\left\{  \sum_{i'\in\mI}  b_{i'} {y}_{i'}(b) \right\} =
{y}_i(b).
\end{equation*}
and
\begin{equation}\label{ygoesinfty}
\lim_{||b || \rightarrow\infty} \sum_{i\in\mI} b_i{y}_i(b) = \infty.
\end{equation}
\end{lemma}
 \begin{proof}
{The optimal value of the assignment problem \eqref{ASMT} is convex as a
function of $b$, since it is the supremum of a set of linear functions. Thus 
the function $b \mapsto \sum_{i\in\mI} b_i {y}_i(b)$, 
a linear combination of convex functions, is also convex.
Further $\sum_{i\in\mI} b_i {y}_i(\tilde{b})$  is a supporting
hyperplane at the point $\tilde{b}$. 
Differentiability can be shown to follow 
from the continuity of
${y}(\tilde{b})$ as a function of $\tilde{b}$, and we next give a detailed
proof} {following the Envelope Theorem~\cite[Chap.
3]{milgrom2004putting}.}

{Since by definition, $x^{\tau}(b)$ is optimal for the assignment problem, we have that
\begin{align*}
\sum_{i\in\mI} b_i {y}_i(b)& = \bE \left[ \sum_{i\in\mI} \sum_{l\in\mL} b_i  p_{il}^\tau x_{il}^{\tau}(b) \right]
=  \bE \left[ \max_{x^\tau\in\mS} \sum_{i\in\mI} \sum_{l\in\mL} b_i  p_{il}^\tau x_{il}^{\tau} \right],
\end{align*}
Letting $b^h= b + e_{i} h$, where $e_i$ is the $i$th unit vector in $\bR^\mI$ and $h>0$ (a symmetric argument holds for $h<0$). We see that the partial derivative with respect to $b_i$ is lower-bounded
\begin{align}
 \sum_{i'\in\mI} \frac{ b^h_{i'} {y}_{i'} ( b^h) - b_{i'} {y}_{i'}(b) }{h}
 \geq & \frac{1}{h} \left\{  \bE \left[ \sum_{i'\in\mI} \sum_{l\in\mL} b^h_{i'}  p_{i'l}^\tau x_{i'l}^{\tau}(b) \right] -  \bE \left[ \sum_{i'\in\mI} \sum_{l\in\mL} b_{i'}  p_{i'l}^\tau x_{i'l}^{\tau}(b) \right] \right\} \label{ypartial} \\
=& \bE \left[  \sum_{l\in\mL}   p_{il}^\tau x_{il}^{\tau}(b) \right]  = {y}_i(b).\notag
\end{align}
The inequality above holds because $x^\tau(b)$ is suboptimal for the parameter choice $b^h$. By the same argument, applied to $b_{i}y_{i}(b)$ instead of  $b_{i'}y_{i'}(b)$ in \eqref{ypartial},  we also have that
\begin{align*}
& \sum_{i'\in\mI} \frac{ b^h_{i'} {y}_{i'} ( b^h) - b_{i'} {y}_{i'}(b) }{h} \leq {y}_i(b^h).
\end{align*}
ny the continuity of ${y}_i(b)$, letting $h\rightarrow 0$ gives that
\begin{equation*}
\frac{d}{d b_i} \sum_{i'\in\mI}  b_{i'} {y}_{i'}(b) = {y}_i(b),
\end{equation*}
as required. }

Since $y_i(b_i,b_{-i})$ is non-zero for
$b_i>0$ it follows that 
\[
\min_{ ||b||=1} \sum_{i\in\mI} b_i{y}_i(b) >0,
\]
and consequently, letting $||b||\to \infty$, we see that \eqref{ygoesinfty}
holds.
 
 \end{proof}

\section{Proof of Propositions \ref{Decomp1} and \ref{NashProp}}

\begin{proof}[Proof of Proposition \ref{Decomp1}]
A Lagrangian of the system problem \eqref{SYS1}, \eqref{SYS2}  can be written as follows
\begin{align}\label{Lag}
L_{sys}(x,y;b)&=\sum_{i\in\mI} U_i(y_i) +\sum_{i\in\mI} b_i  \bE_\tau\! \left[ \sum_{l\in\mL} p_{il}^\tau x_{il}^\tau - y_{i} \right].
\end{align}
Note that we intentionally  omit  the scheduling constraints from  our Lagrangian, and therefore we   must maximize subject to these constraints, (\ref{SYS3}-\ref{SYS4}), when optimizing our Lagrangian.  

Let 
$\mA$ be the set of
variables $x=(x^\tau \in\mS: \tau \in \mT)$
satisfying the assignment constraints~(\ref{SYS3}-\ref{SYS5}).
We see that our Lagrangian problem is separable in the following sense
\begin{subequations}\label{Lsep}
\begin{align}
\max_{\substack{x\in\mA\\ y,z\in\bR^\mI_+; }} L_{sys}(x,y;b) &=\sum_{i\in\mI} \max_{y_i\geq 0}\left\{ U_i(y_i)-b_iy_i \right\} \label{Lsep1}\\
&+  \bE_\tau \left[ \max_{x^\tau \in \mS} \sum_{i\in\mI}\sum_{l\in\mL} b_i p_{il}^\tau x_{il}^\tau \right] \label{Lsep2}
\end{align}
\end{subequations}
Here $\mS$ denotes the set of $x'\in\bR_+^{\mI\times\mL}$ such that for each $i\in\mI$ and $l\in\mL$
\begin{equation*} 
\sum_{l'\in\mL} x'_{il'} \leq 1, \text{ and } \sum_{i'\in\mI} x'_{i'l} \leq 1.
\end{equation*}
We now show that solutions $\tilde{x}$, $\tilde{y}$ and $\tilde{b}$  satisfying the Conditions A and B of our Proposition  are optimal for the  Lagrangian \eqref{Lsep1} and \eqref{Lsep2} when $b=\tilde{b}$.

Firstly, suppose Conditions A and B are satisfied. Assuming Condition A, the following is a straightforward application of Fenchel duality. If $\tilde{b}_i$ a solution to the optimization 
\begin{equation*}
{\rm minimize}\quad {U^*_i(b_i) +  b_i \tilde{y}_i}\quad {\rm over} \quad b_i \geq 0,
\end{equation*}
then, under our expression~\eqref{DU}  for $U^*$, the solution is achieved
when $D_i(\tilde{b}_i)=\tilde{y}_i$ or equivalently when $U_i'(\tilde{y}_i)
= \tilde{b}_i$. Thus it is clear that $\tilde{y}_i$ solves the optimization
\begin{equation*}
 \max_{y_i\geq 0}\{ U_i(y_i) -\tilde{b}_iy_i\}.
\end{equation*}
Hence if Condition A is satisfied, then $\tilde{y}_i$ optimizes \eqref{Lsep1} when we choose $b_i=\tilde{b}_i$.

Secondly,  if $\tilde{x}^\tau$ solves ASSIGNMENT($\tau$,$\tilde{b}$) for each $\tau$, because each maximization inside the expectation \eqref{Lsep2} is an assignment problem, then  \eqref{Lsep2} is maximized by $\tilde{x}$ when we take $b=\tilde{b}$.

These two conditions, Condition A and B, show that the Lagrangian~\eqref{Lag} 
is maximized by $\tilde{x}$ and $\tilde{y}$ with Lagrange multipliers $\tilde{b}$. In addition,  $\tilde{x}$ and $\tilde{y}$  are feasible for the system optimization \eqref{SYS} and hence we have a feasible optimal solution for this Lagrangian problem. But as we  demonstrate in Proposition \ref{AppendixProp} below,  Lagrangian sufficiency still holds for the system problem \eqref{SYS} -- despite the infinite number of constraints. Therefore we have shown a solution to Conditions A and B is optimal for the system problem.

Conversely, we know that strong duality holds for the system optimization \eqref{SYS} --  even with the infinite number of constraints for this optimization (see Theorem \ref{AppendixTheorem} in the Appendix for a proof).  Hence there exists a vector $\tilde{b}$ such that an optimal solution to the system problem is also an optimal solution to the Lagrangian problem when we chose Lagrange multipliers $\tilde{b}$. Thus, an optimal solution to the SYSTEM($U,\mI,\bP_\tau$) must optimize \eqref{Lsep1} and \eqref{Lsep2}, and as discussed these solutions correspond to Conditions A and B. In other words, an optimal solution to the system problem satisfies Conditions A and B with this choice of $\tilde{b}$.   \end{proof}

\begin{proof}[Proof of Proposition \ref{NashProp}]
a) From Theorem \ref{Decomp1}, the Lagrangian of the system problem can be written as follows
\begin{align}\label{LSYS}
L_{sys}(x,y;b)&=\sum_{i\in\mI} U_i(y_i) +\sum_{i\in\mI} b_i  \bE_\tau\! \left[ \sum_{l\in\mL} p_{il}^\tau x_{il}^\tau - y_{i} \right].
\end{align}
Recall from  \eqref{Lsep}, this Lagrangian is separable and is maximized as 
\begin{align*}
\max_{\substack{x\in\mA\\ y,z\in\bR^\mI_+; }} L_{sys}(x,y;b)
=& \sum_{i\in\mI} \max_{y_i\geq 0}\left\{ U_i(y_i)-b_iy_i \right\} +  \bE_\tau \left[ \max_{x^\tau \in \mS} \sum_{i\in\mI}\sum_{l\in\mL} b_i p_{il}^\tau x_{il}^\tau \right]\\
 =& \sum_{i\in\mI} U^*_i(b_i)  + \bE_\tau \left[ \sum_{i\in\mI}b_i  \sum_{l\in\mL} p_{il}^\tau {x}_{il}^{\tau}(b) \right]
= \sum_{i\in\mI} ( U^*_i(b_i) + b_i {y}_i(b) ). 
\end{align*}
In the second equality above, we rearrange the assignment optimization in terms of the click-through rate of each advertiser, ${y}_i(b)$.

Thus the dual of this optimization problem is as required:
\begin{align*}
&\text{Minimize} \quad \sum_{i\in\mI} \left[ U_i^*(b_i)+b_i{y}_i(b) \right]\qquad\text{over} \qquad \quad b_i \geq 0 , \qquad i\in\mI.
\end{align*}

We  analyze this dual problem. We first show that optimization \eqref{SYSdual} is minimized when $0< b_i < \infty$ for each $i\in\mI$. 
We consider the function 
\begin{equation*}
\sum_{i\in\mI} b_i{y}_i(b).
\end{equation*}
With the technical lemma, Lemma \ref{diffsumyprop}, we see that this function is continuous and, for $b_{-i}\neq 0$, differentiable   with $i$th partial derivative given by the continuous function ${y}_i(b)$. 
Further it satisfies~\eqref{ygoesinfty}.
 Thus since $U^*_i(b_i)$ is a positive function, we see that the dual
minimization \eqref{SYSdual} must be achieved by a finite solution $b^*$.
In addition, by definition $D_i(b)=-(U^*_i)'(b)=(U_i')^{-1}(b)$, and so the objective of the dual  is continuously differentiable for $b>0$.  Since 
$D_i(0)=\infty$, the minimum of the dual problem  \eqref{SYSdual} must be achieved by $b^*_i>0$ for each $i\in\mI$.
Now, as the objective of \eqref{SYSdual} is continuously differentiable for $b$ strictly positive, it is minimized iff for each $i\in\mI$
\begin{equation*}
\frac{d U^*_i}{d b_i}(b^*_i) + {y}_i(b^*)=0.
\end{equation*}
Finally, since each function $U^*_i$ is strictly convex, dual objective is strictly convex and so the above minimizer is unique.

b) 
For the Lagrangian for the system problem, \eqref{LSYS}, Strong Duality holds by Theorem \ref{AppendixTheorem}. So, there exist Lagrange multipliers $b^*$, such that
\begin{align}
  \sum_{i\in\mI} \big[ U_i^*(b^*_i)+b^*_i{y}_i(b^*)\big] 
= &  \max_{\substack{x\in\mA\\ y\in\bR^\mI_+; }} L_{sys}(x,y;b^*) 
= \max_{\substack{x\in\mA\\ y\in\bR^\mI_+; }}  \sum_{i\in\mI} U_i(y_i)\notag
\end{align}
where there are feasible vectors $x^*$, $y^*$ achieving the optimum of both maximizations above.  By weak duality it is clear that $b^*$ must be optimal for the dual problem \eqref{SYSdual}. Further, since $x^*$ optimizes the Lagrangian $L_{sys}$ with Lagrange multipliers $b^*$,  it solves the assignment problem, $x^{*\tau}=x^{\tau}(b^*)$. 
   \end{proof}

\section{Lagrangian Optimization} \label{sec:Lagrange}

In this paper, we consider optimization problems that have a potentially infinite number of constraints, in particular, for the system-wide optimization \eqref{SYS}. Thus it is not immediately clear that the Lagrangian approach -- ordinarily applied with a finite number of constraints -- immediately applies to our setting. We demonstrate that certain principle results, namely weak duality, the Lagrangian Sufficiency and strong duality, apply to our setting. These technical lemmas supplement proofs in Propositions  \ref{Decomp1} and \ref{NashProp}.

We consider an optimization of the form
\begin{subequations}\label{LOPT}
\begin{align}
&\rm{Maximize}  &&g(y) &\label{LOPT1}\\
&\rm{subject\;to} && y_{i}\leq \bE_\mu [  x_{i} ],\quad
i=1,...,n, &\label{LOPT2}\\
& && f_j(x(\tau))\leq c_j,\quad \tau\in\mT,\; j=1,...,m,&\label{LOPT3}\\
&\rm{over} &&  y\in\bR^n,\quad x\in\mB(\mT,\bR^n).& \label{LOPT4}
\end{align}
\end{subequations}

In the above optimization, we consider probability space $(\mT,\bP_\mu)$ and measurable random variable $x:\mT\rightarrow \bR^n$. We let $\mB(\mT,\bR^n)$ index the set of Borel measurable functions from $\mT$ to $\bR^n$. We assume that $g:\bR^n\rightarrow\bR$ is a concave function and that $f_j:\bR^n\rightarrow\bR$ is a convex function, for each $j=1,...,m$. We assume the solution to this optimization is bounded above.

Although there are an infinite number of constraints in this optimization, we can define a Lagrangian for this optimization as follows 
\begin{align*}
L(x,y,z;b)&= g(y)+ \sum_{i=1}^n b_i  \bE_\mu [x_i - y_i - z_i]\, .
\end{align*}
Here the Lagrange multipliers $b_i$, $i=1,...n$, can be assumed to be positive, slack variables $z_i$ are added for each constraint \eqref{LOPT2} and the optimization of the Lagrangian is taken over $y_i$ real,  $z_i$ positive and real, and $x_i$ a Borel measurable random variable for $i\in\mI$. We let $\mF$ be the set of $(x,y)$ feasible for the optimization \eqref{LOPT}. 

Weak duality and Lagrangian Sufficiency both hold for this Lagrangian problem.
\begin{proposition}[Weak Duality]\label{AppendixProp}$\:$\\
a) [Weak Duality] For $g^*$ the optimal value of the optimization \eqref{LOPT},
\begin{equation*}
\sup_{\substack{y\in\bR^n,\\ x\in\mB(\mT,\bR^n)}}  L(x,y,z;b)  \geq  g^*.
\end{equation*}
\noindent b) [Lagrangian Sufficiency] If, given some $b$, there exists $x^*\in\mB(\mT,\bR^n)$ and $y^*,z^*\in\bR^n$ that are feasible for the optimization \eqref{LOPT} and maximize the Lagrangian $L(x,y,z;b)$ with $z^*_i := y^*_i-\bE_\mu x^*_i$ then $x^*$, $y^*$, $z^*$ is optimal for \eqref{LOPT}.
\end{proposition}
\begin{proof}
a) Because $\mF$ is a subset of $\mB(\mT,\bR^n) \times \bR^n$, we have
\begin{align*}
\sup_{\substack{y\in\bR^n\!\! ,\, z\in\bR^n_+,\\ x\in\mB(\mT,\bR^n)}} L(x,y,z;b) & \geq \sup_{\substack{(x,y)\in\mF\\ z\in\bR_+^n}} L(x,y,z;b) = g^*.
\end{align*}
This proves weak duality. \\

\noindent b) Now applying this inequality, if a feasible solution optimizes the Lagrangian, then
\begin{align*}
g(y^*) =L(x^*,y^*,z^*;b)=\sup_{\substack{y\in\bR^n\!\! ,\, z\in\bR^n_+,\\ x\in\mB(\mT,\bR^n)}} L(x,y,z;b)\geq g^*. 
\end{align*}
Thus, $(x^*,y^*)$ is optimal for \eqref{LOPT}.
\end{proof}

 For $z\in\bR^\mI$, we use $\mF(z)$ to denote the set of $(x,y)$ satisfying constraints (\ref{LOPT3}-\ref{LOPT4}) and satisfying constraints 
\begin{equation*}
 z_i + y_{i}\leq \bE_\mu [  x_{i} ],\qquad i=1,...,n.
\end{equation*}
Note, $\mF=\mF(0)$. We now show that there exists a Lagrange multiplier $b^*$ where the optimized Lagrangian function also optimizes \eqref{LOPT}.

\begin{theorem}[Strong Duality]\label{AppendixTheorem}
There exists  a $b^*\in\bR_+^n$ such that 
\begin{equation}\label{SDeq}
\max_{(x,y)\in \mF}\; g(y)= \max_{\substack{y\in\bR^n\\x\in\mB(\mT,\bR^n) } }\! g(y)+\sum_{i\in\mI}  b_i^* \bE \left[ x_i-y_i\right].
\end{equation}
In particular, if there exist $(x^*,y^*)\in\mF$ maximizing \eqref{LOPT} then it maximizes \eqref{SDeq}.
\end{theorem}
\begin{proof}
Firstly, since $\mF \subset \mB(\mT,\bR^n)\times  \bR^n$, we proved the weak duality expression
\begin{align}\label{SDleq}
\max_{(x,y)\in \mF}  g(y) & =\max_{(x,y)\in\mF} g(y)+\sum_{i\in\mI} b_i^* \bE \left[ x_i-y_i\right] \leq \max_{\substack{y\in\bR^\mI\\ x\in\mB(\mT,\bR^n) } } g(y)+\sum_{i\in\mI} b_i^* \bE \left[ x_i-y_i\right].
\end{align}
So it remains to show the reverse inequality. We consider the following set
\begin{equation*}
\mC=\{ (z,\gamma)\in\bR^\mI\times\bR : \text{ there exists } (x,y)\in\mF(z)\text{ with } g(y)\geq \gamma \}.
\end{equation*}
We claim that $\mC$ is convex. Take $(z^0,\gamma^0),(z^1,\gamma^1)\in\mC$ and take $(x^0,y^0)\in\mF(z^0)$, $(x^1,y^1)\in\mF(z^1)$ respectively achieving bounds $g(y^0)\geq \gamma^0$ and $g(y^1)\geq \gamma^1$. 
For each term $u=x,y,z,\gamma$ just defined, we correspondingly define $u^q=(1-q) u^0 + q u^1$, for $q\in [0,1]$. 

By concavity of $g$, convexity of $f_j$, $j=1,...,m$, and linearity, we have 
\begin{align*}
g(y^q) &\geq (1-q)g(y^0) + q g(y^1) \geq \gamma^q,\\
f_j(x^{q}(\tau)) &\leq  (1-q)f_j(x^{0}(\tau)) + q f_j(x^{1}(\tau)) \leq c_j, \\
\bE_\mu [  &x^q_{i}-y^q_i ] = (1-q) z_i^0+ qz^1_i=  z_i^q,
\end{align*}
for $\tau\in\mT, j=1,...,m$ and $i=1,...,n$. These above inequalities show that $(z^q,\gamma^q)\in\mC$ and thus our claim is holds: $\mC$ is convex.

Let $\gamma^*=\max_{(x,y)\in\mF} g(y)$. Here we are optimizing over $\mF(z)$ with $z=0$. So, it is clear that $(0,\gamma^*)$ does not belong to the interior of  $\mC$. Thus by the Supporting Hyperplane Theorem \cite{ro97}, there exists a hyperplane through $(0,\gamma^*)$ supporting $\mC$. In other words, there exists a non-zero vector $(b,\phi)\in\bR^\mI\times\bR$ such that
\begin{equation*}
\phi \gamma^* \geq \phi \gamma + b^\T z,
\end{equation*}
for all $(z,\gamma)\in\mC$. Firstly, it is clear that $\phi\geq 0$, otherwise $\gamma^*$ is not maximal for $(x,y)\in\mF$. 

We now claim $\phi\neq 0$. We proceed by contradiction. If $\phi=0$, then $0\geq b^\T z$ for all $(z,\gamma)\in\mC$. But notice, for any $x\in\mB(\mT,\bR_+^n)$, we can choose $y_i\in\bR$ such that $y_i -\bE_\mu [  x_{i} ]=b_i$, thus for this choice of $(x,y)$ we have $z=b$. Thus, $b^\T z =b^\T b  > 0$, and so we have a contradiction. It must be that $\phi>0$. 

As $\phi>0$, we can define $b^*= (b_i/\phi: i\in\mI)$. Since for each $(x,y)\in\mB(\mT,\bR^n)\times \bR^\mI$, if we set $z_i=\bE_{\mu} \left[ x_i-y_i\right]$ and $\gamma=g(y)$ then we have $(z,\gamma)\in\mC$. With this we have
\begin{align*}
\max_{(x',y')\in \mF}  g(y') = \gamma^* \geq \gamma + b^{*\T} z= g(y) + \sum_{i\in\mI} b^*_i\bE \left[ x_i-y_i\right]
\end{align*}
Thus, maximizing over $x\in\mB(\mT,\bR^n)$ and $y\in\bR^\mI$, we have 
\begin{equation}\label{SDgeq}
\max_{(x,y)\in \mF}  g(y)\geq \!\!\!\max_{\substack{y\in\bR^n\\x\in\mB(\mT,\bR^n) } }\! g(y)+\sum_{i\in\mI}  b_i^* \bE \left[ x_i-y_i\right].
\end{equation}
Together \eqref{SDleq} and \eqref{SDgeq} give the required equality \eqref{SDeq}. In addition, given \eqref{LOPT} has a finite optimum, inequality \eqref{SDgeq} can only hold when $b^*\geq 0$.

Finally, if $(x^*,y^*)\in\mF$ are optimal for  \eqref{LOPT} then equality \eqref{SDeq} implies
\begin{equation*}
g(y^*) \geq g(y^*) + \sum_{i\in\mI}  b_i^* \bE \left[ x^*_i-y^*_i\right].
\end{equation*}
However, the feasibility of $(x^*,y^*)$ and positivity of $b^*$ implies 
\begin{equation*}
\sum_{i\in\mI} b_i^* \bE \left[ x^*_i-y^*_i\right] \geq 0. 
\end{equation*}
So we see these two inequalities imply complementary slackness  $b_i^* \bE \left[ x^*_i-y^*_i\right] =0$ and that $(x^*,y^*)\in\mF$ maximizing \eqref{LOPT} also maximizes the right-hand side of \eqref{SDeq}.
\end{proof}

\section{Dynamics} \label{sec:dyn}

We finally note a result used in the proof of Theorem~\ref{CONV}.

\begin{lemma}\label{Lemma:dyn} 
If $b_i < U_i'(y_i)$ then there exists $\delta>0$ such
that
$u_i(b) <  u_i(b_i + \epsilon, b_{-i})$ for all $\epsilon \in (0,
\delta)$. Similarly if $b_i > U_i'(y_i)$ then there exists $\delta>0$ such
that
$u_i(b) >  u_i(b_i + \epsilon, b_{-i})$ for all $\epsilon \in (0,
\delta)$.
\end{lemma}

\begin{proof}
From the definitions~\eqref{bidpriceprimal1} and~\eqref{rewardsprimal},
and from an application of  the mean value theorem for some $\tilde{b}$
satisfying $ y_i(b) < \tilde{b} < y_i(b_i + \epsilon, b_{-i}) $, we have
that 
\begin{align*}
 & u_i(b_i + \epsilon, b_{-i}) - u_i(b)  \\
&= U_i(y_i(b_i + \epsilon, b_{-i}))
-  U_i(y_i(b)) - \int_0^{b_i + \epsilon} [y_i(b_i + \epsilon, b_{-i}) -
   y_i(b', b_{-i})]
  db'_i
+ \int_0^{b_i} [y_i(b) - y_i(b', b_{-i})]
  db'_i  \\
&= \left(U'_i(y_i(\tilde{b}))-b_i \right) [y_i(b_i + \epsilon, b_{-i}) - y_i(b)]  
- \int_{b_i}^{b_i + \epsilon} [y_i(b_i + \epsilon, b_{-i}) - y_i(b',
  b_{-i})] db'_i . 
\end{align*}
But 
\[
0 < \int_{b_i}^{b_i + \epsilon} [y_i(b_i + \epsilon, b_{-i}) - y_i(b',
  b_{-i})] db'_i < \epsilon [y_i(b_i + \epsilon, b_{-i}) - y_i(b)] 
\]
and thus the result follows for $\delta$ sufficiently small. 
\end{proof}